%% file: n_distance.tex
\documentclass{article}

\usepackage{microtype}
\usepackage{subfigure}
\usepackage{booktabs} 

\usepackage[accepted,nohyperref]{icml2019}

\usepackage{ifpdf}
\usepackage[pdftex]{graphicx}
\usepackage{epstopdf}
\usepackage{natbib}
\usepackage{amsmath}
\usepackage{amssymb}
\usepackage{amsthm}
\usepackage{amsfonts}
\usepackage{verbatim}
\usepackage{algorithm}
\usepackage{algorithmic}
\usepackage{array}
\usepackage{url}
\usepackage{color}
\usepackage{mathtools}
\usepackage{float}
\usepackage{dblfloatfix} 
\usepackage[dvipsnames]{xcolor}
\usepackage{mathabx}

\newtheorem{theorem}{Theorem}
\newtheorem{definition}{Definition}
\newtheorem{assumption}{Assumption}
\newtheorem{lemma}{Lemma}
\newtheorem{corollary}{Corollary}

\newtheorem{remark}{Remark}

\begin{document}

\twocolumn[
\icmltitle{Tractable $n$-Metrics for Multiple Graphs}




\begin{icmlauthorlist}
\icmlauthor{Sam Safavi}{to}
\icmlauthor{Jos\'e Bento}{to}
\end{icmlauthorlist}

\icmlaffiliation{to}{Department of Computer Science, Boston College, Chestnut Hill, MA, USA}

\icmlcorrespondingauthor{Jos\'e Bento}{jose.bento@bc.edu}

\icmlkeywords{Machine Learning, ICML}

\vskip 0.3in
]



\printAffiliationsAndNotice{}  

\begin{abstract}
Graphs are used in almost every scientific discipline to express relations among a set of objects. Algorithms that compare graphs, and output a closeness score, or a correspondence among their nodes, are thus extremely important. Despite the large amount of work done, many of the scalable algorithms to compare graphs do not produce closeness scores that satisfy the intuitive properties of metrics. This is problematic since non-metrics are known to degrade the performance of algorithms such as distance-based clustering of graphs \cite{bento2018family}. On the other hand, the use of metrics increases the performance of several machine learning tasks \cite{indyk1999sublinear,clarkson1999nearest,angiulli2002fast,ackermann2010clustering}. In this paper, we introduce a new family of multi-distances (a distance between more than two elements) that satisfies a generalization of the properties of metrics to multiple elements. 
In the context of comparing graphs, we are the first to show the existence of multi-distances that simultaneously incorporate the useful property of \emph{alignment consistency} \cite{nguyen2011optimization}, and a generalized metric property. Furthermore, we show that these multi-distances can be relaxed to convex optimization problems, without losing the generalized metric property.
%

\end{abstract}

\vspace{-1cm}
\section{Introduction}\label{sec:intro}

A canonical way to check if two graphs $G_1$ and $G_2$ are similar, is to try to find
a map $P$ from the nodes of $G_2$ to the nodes of $G_1$ such that, for many pairs of nodes in $G_2$, their images in $G_1$ through $P$ have the same connectivity relation (connected/disconnected)~\cite{deza2009encyclopedia}. For {equal-sized} graphs, this can be formalized as 
\begin{equation}\label{eq:chem_dist}
d(G_1,G_2) \hspace{-0.9mm} \triangleq \hspace{-0.9mm}\min_{P}\{ \vvvert A_1 - PA_2P^\top\vvvert \hspace{-1mm}=\hspace{-1mm} \vvvert A_1P - PA_2\vvvert  \},
\end{equation}
where $A_1$ and $A_2$ are the adjacency matrices of $G_1$ and $G_2$, $P$ and its transpose $P^\top$ are permutation matrices, and, here, $\vvvert\cdot\vvvert$ is the Frobenius norm. A map $P^*$ that minimizes \eqref{eq:chem_dist} is called an optimal \emph{alignment} or \emph{match} between $G_1$ and $G_2$. 
If $d(G_1,G_2)$ is small (resp. large), we say $G_1$ and $G_2$ are topologically similar (resp. dissimilar). Computing $d$, or $P^*$, is hard \cite{klau2009new}. Determining if $d(G_1,G_2)=0$, which is the graph isomorphism problem, is not known to be in P, or in NP-hard \cite{babai2016graph}. 

%
%
Scalable alignment algorithms, which find an approximation $P$ to an optimal alignment  $P^*$, or find a solution to a tractable variant of \eqref{eq:chem_dist}, e.g., \cite{klau2009new,bayati2013message,singh2008global,el2015natalie}, have mostly been developed with no concern {as to} whether the closeness score $d$ obtained from the  alignment $P$, e.g., computed via $d(G_1,G_2) = \|A_1P - PA_2\|$, results in a non-metric. An exception is the recent work in \cite{bento2018family}. Indeed for the methods in, e.g.,  \cite{klau2009new,bayati2013message,singh2008global,el2015natalie},  the work of \cite{bento2018family} shows that one can find two graphs that are individually similar to a third one, but not similar to each other, according to $d$. Furthermore, \cite{bento2018family} shows how the lack of the metric properties can lead to a degraded performance in a clustering task to automatically classify different graphs into the categories: Barabasi Albert, Erdos-Renyi, Power Law Tree, Regular graph, and Small World. At the same time, the metric properties allow us to solve several machine learning tasks efficiently \cite{indyk1999sublinear,clarkson1999nearest,angiulli2002fast,ackermann2010clustering}, as we now illustrate.

{\bf Diameter estimation:} Given a set $S$ with $|S|$ graphs, we can compute the maximum diameter $\Delta \triangleq \max_{G_1,G_2 \in S} d(G_1,G_2)$ by computing $\binom{|S|}{2}$ distances. However, if $d$ is a metric, we know that there are at least $\Omega(|S|)$ pairs of graphs with $d \geq \Delta /2$. Indeed, if $d(G^*,G_*) = \Delta$, then, by the triangle inequality, for any $G \in S$, we cannot have both $d(G^*,G) < \Delta/2$ and $d(G_*,G) < \Delta/2$ . Therefore, if we evaluate $d$ on random pairs of graphs, we are guaranteed to find an $1/2$-approximation of $\Delta$ with only $\mathcal{O}(|S|)$ distance computations, on average.

Being able to compare two graphs is important in many fields such as biology~\cite{kalaev2008networkblast,zaslavskiy2009global,kelley2004pathblast,weskamp2007multiple}, object recognition~\cite{conte2004thirty},
dealing with ontologies ~\cite{hu2008matching,wang2016topicpanorama}, computer vision~\cite{conte2004thirty}, and social networks~\cite{zhang2015multiple}, and graph clustering \cite{ma2016multi}, to name a few. In many applications, however, one needs to jointly compare multiple graphs. This is the case, for example, in aligning protein-protein interaction networks~\cite{singh2008global}, recommendation systems, in the collective analysis of networks, or 
in the alignment of graphs obtained from brain MRI~\cite{papo2014complex}.  
 
The problem of jointly comparing $n$ graphs, $n \geq 3$, is harder, and has been studied far less than when $n=2$. Examples and applications include \cite{pachauri2013solving,douglas2018metrics,yan2015consistency,gold1996graduated,hu2016distributable,park2016encouraging,huang2013consistent,sole2011models,williams1997multiple,hashemifar2016joint,heimann2018node,nassar2017multimodal,feizi2016spectral,chen2014near}.

\vspace{0mm}
Consider the search for a function $d(G_1,...,G_n)$ that scores how close $G_1,...,G_n$ are. New questions arise when $n\geq 3$:
\vspace{-0.5cm}
\begin{enumerate}
\item If $d$ produces alignments between each pair of graphs in $\{G_1,\dots,G_n\}$,  should these alignments be related? What properties should they satisfy? \label{quest_1}
\vspace{-2mm}
\item Should $d$ satisfy similar properties to that of a metric? What properties?\label{quest_2}
\vspace{-2mm}
\item Is it possible to find a $d$ that is tractable? Is it possible to impose on $d$ the properties from \ref{quest_1} and \ref{quest_2} above without losing tractability?
\end{enumerate}
\vspace{-0.2cm}

Multi-graph alignment scores, are important in many applications. For example, many problems require clustering using $n$th order interaction \textcolor{black}{\citep{leordeanu2012efficient}}, i.e., clustering based on the similarity of groups of $n$ elements, not just groups of two elements, as in spectral, or hierarchical clustering. Furthermore, having a score function $d(G_1,\dots,G_n)$ with some form of generalized metric property can have advantages, similar to what \cite{bento2018family} showed for metrics (cf. Section \ref{sec:n-metrics}).

In this paper, we are the first to provide a family of similarity scores for jointly comparing multiple graphs that simultaneously (a) give intuitive joint alignments between graphs, (b) satisfy similar properties to those of metrics, and (c) can be computed using convex optimization methods.

%
%
%
\section{Related work}
Consider three graphs $G_1$, $G_2$, and $G_3$, and three permutation matrices $P_{1,2}$, $P_{2,3}$ and $P_{1,3}$, where the map $P_{i,j}$ is an alignment between the nodes of graphs $G_i$ and $G_j$. 
An intuitive property that is often required for these alignments is that if $P_{1,2}$ maps (the nodes of) $G_1$ to $G_2$, and if $P_{2,3}$ maps $G_2$ to $G_3$, then $P_{1,3}$ should map $G_1$ to $G_3$. Mathematically, $P_{1,3} = P_{1,2}P_{2,3}$. This property is often called \emph{alignment consistency}. Papers that enforce this constraint, or variants of it, include \cite{huang2013consistent,pachauri2013solving,chen2014near,yan2015matrix,yan2015consistency,zhou2015multi,hu2016distributable}. Most of these papers focus on computer vision, i.e., the task of producing alignments between shapes, or reference points among different figures, although most of the ideas can be easily adapted to aligning graphs. The proposed alignment algorithms are not all equally easy to solve, some involve convex problems, others involve non-convex or integer-valued problems. None of these works care about the alignment scores satisfying metric-like properties.

There are several papers that propose procedures for generating multi-distances from pairwise distances, and  prove that these multi-distances satisfy 
intuitive generalizations of the metric properties to $n\geq 3$ elements.
These allow us to use the existing works on two-graph comparisons to produce distances between multiple graphs. The simplest method is to define $d(G_1,\dots,G_n) = \sum_{i,j \in [n]} d(G_i,G_j)$. The problem with this approach is that if $d(G_i,G_j)$ also produces an alignment $P_{i,j}$, e.g., in \eqref{eq:chem_dist}, these alignments are unrelated, and hence do not satisfy consistency constrains that are usually desirable. An approach studied by \cite{kiss2016generalization} is to define $d(G_1,\dots,G_n) = \min_{G} \sum_{i\in[n]} d(G_i,G)$. If each $d(G_i,G)$ also produces an alignment $P_i$, and if we define $P_{i,j} = P_i P^\top_j$, then $\{P_{i,j}\}$ is a set of alignments that satisfy the aforementioned consistency constraint. The problem with this approach is that it tends to lead to computationally harder problems, even after several relaxations are applied (cf. Fermat distance in Section \ref{sec:n-metrics}). A few other works that study metrics and their generalizations are \cite{kiss2016generalization, martin2011functionally,akleman1999generalized}. 

The work of \cite{bento2018family} defines a family of metrics for comparing two graphs. Several metrics in this family are tractable, or can be reduced to solving a convex optimization problem. However, \cite{bento2018family} does not consider comparing $n\geq 3$ graphs.
We refer the reader to~\cite{khamsi2015generalized} that surveys generalized metric spaces, and~\cite{deza2009encyclopedia} that provides an extensive review of many distance functions along with their applications in different fields, and, in particular, discusses the generalizations of the concept of metrics in different areas such as topology, probability, and algebra. The authors in~\cite{deza2009encyclopedia} also discuss several distances for comparing two graphs, most of which are not tractable.

\section{Notation and preliminaries}\label{sec:prel}

\begin{table}[]
\begin{tabular}{|ll|ll|}
\hline
$G_i$        & $i$th graph               & $P_{i,j}$                 & Alig. of  $G_i$ and $G_j$       \\ \hline
$A_i$        & Adj. mat. of $G_i$ & $\mathcal{P}$             & Set of alig. mats.           \\ \hline
$n$          & \# of graphs              & $d$                       & Dist. among $n$ graphs         \\ \hline
$m$          & \# of nodes               & $\Omega$                  & Set of adj. mats. \\ \hline
$s$          & Alig. score           & $S$                       & Set of sets of alig. mats.   \\ \hline
${\bf P}$ & Mat. of $\{P_{i,j}\}$   & $\|| \cdot\||$ & Mat. norm              \\ \hline
$\|\cdot\|$ & Vec. norm  & {\bf tr} & Trace              \\ \hline
\end{tabular}
\caption{Summary of main notation used.}
\vspace{-0.8cm}
\end{table}

We focus on comparing graphs of equal size.
A canonical way to deal with graphs with different sizes is to add dummy nodes to make them equal-sized. Many
applied papers, e.g., \textcolor{black}{\cite{zaslavskiy2009global,zaslavskiy2009path,narayanan2011link,zaslavskiy2010many,zhou2012factorized,gold1996softmax,yan2015general,sole2010graduated,yan2015consistency}}, follow this approach.

Comparing equal-sized graphs, without adding dummy nodes is still important. One application in computer vision is to establish a correspondence among the nodes of $n$ graphs, each representing a geometrical relation among $m$ special points in $n$ images of the same object. The user (or detection algorithm), by design, finds the same number, $m$, of special points in each image. \textcolor{black}{See, e.g., the numerical experiments in \cite{hu2016distributable,shen2015person}. Other papers that only consider equal-sized graphs include: \cite{lyzinski2016graph,pachauri2013solving}.}
We also point the reader to the remark on comparing graphs of unequal size at end of Section \ref{sec:tractabilty}.

Let $[m] = \{1,\dots,m\}$. A graph, $G = (V\equiv [m],E)$, with node set $V$ and edge set $E$, is represented by a matrix,~$A$, whose entries are indexed by the nodes  in $V$. We denote the set that contains all such matrices by $\Omega \subseteq \mathbb{R}^{m \times m}$. E.g.,~$\Omega$ can be the set of adjacency matrices, or of the matrices containing hop-distances between all pairs of nodes.

Consider a set of $n$ graphs, $\mathcal{G} = \{G_1,G_2,\ldots,G_n\}$. Given two graphs, $G_i = (V_i,E_i)$ and $G_j=(V_j,E_j)$, from the set~$\mathcal{G}$, we denote {a pairwise matching matrix} between $G_i$ and $G_j$ by $P_{i,j}$. The rows and columns of $P_{i,j}$ are indexed by the nodes in $V_i$ and $V_j$, respectively. Note that we can extract a relation between $E_i$ and $E_j$, from a relation between $V_i$ and $V_j$. We denote the set of all pairwise matching matrices by $\mathcal{P} = \{ \{P_{i,j}\}_{i,j \in [n]} : P_{i,j} \in \mathbb{R}^{m \times m}\}$.  For example, $\mathcal{P}$ might be all \emph{permutation matrices} on~$m$ elements.

Let $1{:}n$ denote the sequence $1,\ldots, n$. For $A_1,\dots,A_n \in \Omega$, we denote the {ordered} sequence $(A_1,\ldots,A_n)$ by $A_{1:n}$. The notation $A_{1:n,n+1}^i$ corresponds to the sequence $A_{1:n}$, in which the $i$th element, $A_i$, is removed and replaced by $A_{n+1}$. If~$\sigma$ is a permutation, i.e., a bijection from $1{:}n$ to $1{:}n$ such that~$\sigma(i) = j$, then~$A_{\sigma(1:n)}$ represents a sequence, whose $i$th element is $A_j$. 
{In this paper, we use $\Vert \cdot \Vert$ and $\vvvert \cdot \vvvert$ to denote vector norms and matrix norms, respectively.}
We now provide the following definitions that will be used in the next sections of the paper. In what follows, equality of graphs means that they are isomorphic.

\begin{definition}\label{def1}
A map $d: {\Omega}^2\mapsto {\mathbb{R}}$, is a metric, if and only if, for all $A,B,C \in \Omega$: \emph{(i)} $d(A,B) \geq 0$; \emph{(ii)} $d(A,B) =0, \text{ iff } A=B$; \emph{(iii)} $d(A,B)=d(B,A)$; and \emph{(iv)} $d(A,C) \leq d(A,B) + d(B,C)$.
%
\end{definition}
%
%
\begin{definition}\label{def1.1}
A map $d: {\Omega}^2\mapsto {\mathbb{R}}$ is a pseudometric, if and only if it satisfies properties \emph{(i), (iii)} and \emph{(iv)} in \emph{Definition \ref{def1}}, and if $d(A,A)=0 \;\; \forall A \in \Omega$.
%
%
%
\end{definition}
\vspace{0mm} 
\noindent Given a  pseudometric $d$ on two graphs, we define the equivalence relation $\sim_d$ in $\Omega$  as $A \sim_d B$ if and only if $d(A,B) = 0$.
Using the fact that $d$ is a pseudometric, it is immediate to {verify} that the binary relation $\sim_d$ satisfies {\emph{reflexivity}}, \emph{symmetry} and {\emph{transitivity}}. We denote by $\Omega'=\Omega \backslash \sim_d$ the quotient space $\Omega$ {modulo} $\sim_d$, and, for any $A\in \Omega$, we let $[A]\subseteq \Omega$ denote the equivalence class of $A$. Given $A_{1:n}$, we let $[A]_{1:n}$ denote $ ([A_1],\dots,[A_n])$, an ordered set of sets.

\begin{definition}\label{def2}
A map $s: {\Omega}^2 \times \mathcal{P} \mapsto {\mathbb{R}}$ is called a $P$-score, if and only if, $\mathcal{P}$ is closed under inversion, and for any $P, P' \in \mathcal{P}$, and $A, B, C \in \Omega$, $s$ satisfies the properties:\vspace{0mm}
\begin{align}
&s(A,B,P) \geq 0\label{ds1},\\
&s(A,A,I) =0\label{s3},\\
&{s(A,B,P) = s(B,A,P^{-1})}, \label{s1}\\
&s(A,B,P) + s(B,C,P') \geq s(A,C,PP')\label{s2}.
\end{align}
\end{definition}
\vspace{0mm} 
\noindent For example, if $\mathcal{P}$ is the set of permutation matrices, and $\vvvert \cdot \vvvert$ is an element-wise matrix $p$-norm, then $s(A,B,P) = \vvvert AP-BP\vvvert$ is a $P$-score.
\begin{definition}[\hspace{-0.1mm}\cite{bento2018family}]
The SB-distance function induced by the norm $\vvvert \cdot \vvvert:\mathbb{R}^{m \times m}\mapsto \mathbb{R}$, the matrix $D \in \mathbb{R}^{m \times m}$, and the set $\mathcal{P}\subseteq\mathbb{R}^{m \times m}$ is the map $d_{SB}: \Omega^2 \mapsto \mathbb{R}$, such that
\vspace{0mm} 
\begin{equation*}\label{eq:SB}
d_{SB}(A,B) = {\min_{P \in \mathcal{P}} \vvvert AP - PB\vvvert} + \text{tr}(P^\top D).
\end{equation*}
\end{definition}
\vspace{0mm} 
\noindent The authors in \cite{bento2018family}, prove several conditions on $\Omega$, $\mathcal{P}$, the norm $\vvvert \cdot \vvvert$, and the matrix $D$, such that $d_{SB}$ is a metric, or a pseudometric. For example, if $\vvvert \cdot \vvvert$ is {an arbitrary entry-wise or operator norm, $\mathcal{P}$ is the set of $n \times n$ doubly stochastic matrices, $\Omega$ is the set of symmetric matrices, and $D$ is a \textit{distance matrix}, then $d_{SB}$ is a pseudometric.}

%

%
%

\section{{$n$-metrics} for multi-graph {alignment}}\label{sec:n-metrics}

One can generalize the notion of a {(pseudo)} metric to $n\geq~3$ elements. 
To this aim, we consider the following definitions.
\begin{definition}\label{def5}
A map $d: {\Omega}^n\mapsto \mathbb{R}$, is an $n$-metric, if and only if, for all $A_1,\ldots, A_n \in \Omega$,
\begin{align}
	&d(A_{1:n}) \geq 0,\label{h1}\\
	&d(A_{1:n}) = 0, \text{ iff } A_1   =\ldots = A_n, \label{h2-n}\\
	&d(A_{1:n}) = d(A_{\sigma(1:n)}),\label{h3}\\
	&d(A_{1:n}) \leq \textstyle\sum_{i=1}^{n} d(A^i_{1:n,n+1}).\label{h4}
\end{align}
\end{definition}
\noindent According to Definition \ref{def5}, a $2$-metric is a metric as per Definition \ref{def1}. {In the sequel, we refer to properties \eqref{h1}, \eqref{h2-n}, \eqref{h3}, and \eqref{h4}, as non-negativity, identity of indiscernibles, symmetry, and generalized triangle equality (GTI), respectively.}

\begin{definition}
A map $d: {\Omega}^n\mapsto \mathbb{R}$, is a pseudo $n$-metric, if and only if it satisfies properties \eqref{h1}, \eqref{h3} and \eqref{h4}, and for any $A \in \Omega$, {$d$ satisfies the property of \emph{self-identity}}
\vspace{-0.2cm}
\begin{equation}\label{eq:self_identity}
d(A,\cdots,A)=0.
\end{equation}
\end{definition}

{\bf Revisiting diameter estimation:} $n$-metrics have several advantages over non-$n$-metrics. For $n=2$, this is shown by \cite{bento2018family} and references therein: metrics allow several ML algorithms to finish faster, and improve the accuracy in tasks such as clustering graphs. Some of these advantages also extend to $n>2$. For example, it is  straightforward to see that, if we generalize the  diameter estimation problem in Sec. \ref{sec:intro} to $n=3$, we can compute a $1/3$-approximation of
$\max_{G_1,G_2,G_3 \in S} d(G_1,G_2,G_3)$ in expected time $O(n^2)$, compared to $O(n^3)$ for a non-$n$-metric. Considering the runtime of distance-based clustering using \emph{$n$th order interaction} \citep{7582510}, and
just like for $n=2$, $n$-metrics, $n>2$, also improve runtime, because the GTI lets us avoid dealing with all $n$-distances.

We now define two functions that satisfy the properties of {(pseudo)} $n$-metrics.

\subsection{A first attempt: Fermat distances}
\begin{definition}\label{def6}
Given a map $d:\Omega^2\mapsto {\mathbb{R}}$, the 
Fermat distance function induced by $d$, is the map $d_F:\Omega^n \mapsto \mathbb{R}$, defined by
\vspace{-0.5cm}
\begin{equation}\label{eq:fermat}
	d_F(A_{1:n}) = \min\limits_{B \in \Omega} \sum_{i=1}^{n} d(A_i,B).
\end{equation}
\end{definition}
\vspace{-2mm}
{In the context of multiple graph alignment, $d$ is an alignment score between two graphs, and $d_F$ aims to find a graph, represented by $B$, that aligns well with all the graphs, represented by $A_{1{:}n}$.~Thus,  $d_F(A_{1:n})$ can be interpreted as an alignment score computed as the sum of alignment scores between each $A_i$ and $B$. If we think of $A_{1:n}$ as a cluster of graphs, we can think of $B$ as its center.}

\begin{theorem}\label{thm1}
	If $d$ is a {pseudometric}, then the Fermat distance function induced by $d$ is a {pseudo $n$-metric}.
\end{theorem}

The proof of Theorem \ref{thm1} is a direct adaptation of the one in~\cite{kiss2016generalization}, and is included in Appendix \ref{sec:proof_thm1} for completeness. 

For example, the Fermat distance function induced by an SB-distance function with a distance matrix $D=0$~is
\vspace{-3mm}
\begin{equation*}
d_F(A_{1:n}) =\hspace{-0.3cm} \min_{\substack{B \in \Omega,\{P_i\}\in \mathcal{P}^n}}  \sum^n_{i=1}  {\vvvert A_iP_i - P_iB \vvvert} .
\end{equation*}
Despite its simplicity, the above optimization problem 
is not easy to solve in general, {even when} it is a continuous smooth optimization problem. For example, if~$\mathcal{P}$ is the set of doubly stochastic matrices, $B$ is the set of real matrices with entries in $[0,1]$, and $\vvvert \cdot \vvvert$ is the Frobenius norm, the problem is non-convex due to the product $PB$ that appears in the objective function. The potential complexity of computing $d_F$ motivates the following alternative definition.

\subsection{A better approach: $\mathcal{G}$-align distances}

\begin{definition}\label{def7}
Given a map $s:\Omega^2\times \mathcal{P} \mapsto \mathbb{R}$, the $\mathcal{G}$-align distance function induced by $s$, is the map $d_\mathcal{G}:\Omega^n \mapsto~{\mathbb{R}}$, defined by\vspace{-2mm}
\begin{equation}\label{eq:galign}
d_\mathcal{G}(A_{1:n})=\min\limits_{P \in S }\frac{1}{2}\sum\limits_{{{i,j \in [n]}}
} s(A_i,A_j,P_{i,j}),\vspace{-3mm}
\end{equation}
where\vspace{-1mm}
\begin{align}\label{eq:def_of_S}
&S\hspace{-0.cm}=\hspace{-0.cm} \{\{P_{i,j}\}_{i,j\in[n]}\hspace{-0.1cm}:\hspace{-0.cm} P_{i,j}\in \mathcal{P},\hspace{-0.cm}\forall i,j \in [n],~P_{i,k}P_{k,j} = P_{i,j},\nonumber\\ &\forall i,
j,k \in [n], P_{i,i}=I, \forall i\in[n] \}.
\end{align}
\end{definition}
\vspace{-3mm}
\begin{remark}
From the definition of $S$, it is implied that $I \in \mathcal{P}$ and that, if $P\in S$, then $P_{i,j}P_{j,i} =P_{i,i} =I\Leftrightarrow (P_{i,j})=(P_{j,i})^{-1} \forall i,j \in [n]$, hence $\{P_{i,j}\}$ are invertible.
\end{remark}
\begin{remark}
In \eqref{eq:def_of_S}, we refer to the property $P_{i,j}P_{j,k} = P_{i,k}, \forall i,j,k \in [n]$, as \emph{the alignment consistency} {of $P\in S$}.
\end{remark}
The following Lemma, provides an alternative definition for the $\mathcal{G}$-align distance function.
\begin{lemma}\label{th:d_G_compact_sum}
If $s$ is a $P$-score, then
\begin{equation}\label{eq:th:d_G_compact_sum}
d_\mathcal{G}(A_{1:n}) = \min_{P\in S} \sum_{i,j\in[n],~i<j} s(A_i,A_j,P_{i,j}).
\end{equation}
\end{lemma}
\vspace{-5mm}
\begin{proof}
\begin{align}
&~~\sum_{i,j\in[n]} s(A_i,A_j,P_{i,j}) =  
\sum_{i\in [n]} s(A_i,A_i,P_{i,i})  + \nonumber\\
&
\sum_{i,j\in[n]:~i<j} (s(A_i,A_j,P_{i,j}) + s(A_j,A_i,P_{j,i})).
\end{align}
If $P\in S$, then $P_{i,i}=I$ and $P_{j,i} = (P_{i,j})^{-1}$. Thus, since $s$ is a $P$-score, $s(A_i,A_i,P_{i,i}) = s(A_i,A_i,I) = 0$, by property \eqref{s3}, and $s(A_j,A_i,P_{j,i}) = s(A_i,A_j,P_{i,j})$, by property \eqref{s1}. Therefore, 
\begin{equation*}
\sum_{i,j\in[n]} s(A_i,A_j,P_{i,j}) = 2\sum_{i,j\in[n],~i<j} s(A_i,A_j,P_{i,j}),
\end{equation*}
and the proof follows.
\end{proof}

\vspace{-2mm}
Note that, if $s(A,B,P) = {\vvvert AP-PB \vvvert}$, for some element-wise matrix norm, $n=2$, and $\mathcal{P}$ is the set of permutations on $m$ elements, then according to Lemma \ref{th:d_G_compact_sum}, $d_\mathcal{G}(A,B) = d_{SB}(A,B)$, for $D=0$. In general, we can define a generalized SB-distance function induced by a matrix $D$, a  set $\mathcal{P}\subseteq \mathbb{R}^{m \times m}$ and a map $s:\Omega^2\times\mathcal{P}\mapsto \mathbb{R}$ as 
\begin{equation}\label{SB2}
d_{SB}(A,B)= \min_{P\in \mathcal{P}} s(A,B,P) + \text{tr}(P^\top D),
\end{equation}
and investigate the conditions on $s$, $\mathcal{P}$ and $D$, under which \eqref{SB2} represents a {(pseudo)} metric.

{The following lemma leads to an equivalent definition for the $\mathcal{G}$-align distance function, which, among other things, reduces the optimization problem in \eqref{eq:galign}, to finding $n$ different matrices rather that $n^2-n$ matrices that need to satisfy the alignment consistency.} 

\begin{lemma}\label{lem1}
If $S'=\{\{P_{i,j}\}_{i,j\in[n]}: P_{i,j}\in \mathcal{P} \text{ and } P_{i,j} = Q_{i}(Q_{j})^{-1},\forall i,j \in [n], 
\text{ for some matrices } \{Q_{i}\} \subseteq \mathcal{P}\}$, then $S' = S$.
\end{lemma}
\vspace{-4mm}
\begin{proof}
We first prove that $S \subseteq S'$. 
Let $P \in S$. Define $Q_i = P_{i,n} \in \mathcal{P}$ for all $i \in [n]$. If $i,j \in [n-1]$, then, by definition, $P_{i,j} = P_{i,n}P_{n,j} = P_{i,n}(P_{j,n})^{-1} = Q_i (Q_j)^{-1}$. This proves that $P \in S'$.
\vspace{-1mm}

We now prove that $S' \subseteq S$. 
Let $P\in S'$. For any $i,j,k\in[n]$, we have
$P_{i,k}P_{k,j} = Q_{i} (Q_{k})^{-1} Q_{k} (Q_{j})^{-1} = 
Q_{i} (Q_{j})^{-1} = P_{i,j}$. It also follows that $P_{i,j} = Q_i(Q_j)^{-1} = (Q_j(Q_i)^{-1})^{-1} = (P_{j,i})^{-1}$, and $P_{i,i} = Q_i (Q_i)^{-1} =I$. Therefore, $P \in S$.
\end{proof}
\vspace{-2mm}
We complete this section with the following theorem, whose detailed proof is provided in Appendix~\ref{sec:proof_thm2}.

\begin{theorem}\label{thm2}
If $s$ is a $P$-score, then the $\mathcal{G}$-align function induced by $s$ is a pseudo $n$-metric.
\end{theorem}
%
%

\textcolor{black}{In Appendix \ref{sec:ortho_matrices},  we discuss the special case of $\mathcal{P}$ being the set of orthogonal matrices. In this case, we can simplify both eq. \eqref{eq:fermat}, and eq. \eqref{eq:galign}, and compute them efficiently.}

\section{$n$-metrics on quotient spaces}\label{sec:quotient}
The theorems in Section \ref{sec:n-metrics} are stated for pseudometrics. However, it is easy to obtain an $n$-metric from a pseudo $n$-metric for both $d_F$ and $d_\mathcal{G}$ using {quotient spaces}. In these spaces, \eqref{h2-n} holds {almost trivially} (with $A_i$ replaced by its equivalent class $[A_i]$), and the important question is  whether the {equivalent classes of graphs} are meaningful and useful. The proofs for the theorems in this section are Appendices  \ref{app:proofs_for_quotient_spaces_1} and \ref{app:proofs_for_quotient_spaces_2}.

\begin{theorem}\label{th:quotient_1}
Let $d$ be a pseudometric for two graphs, $d_F$ be the
Fermat distance function for $n$ graphs induced by $d$, and
$\Omega' = \Omega \backslash \sim_d$. 
Let $d'_F:\Omega'^n \mapsto \mathbb{R}$ be such that
\begin{equation}\label{eq:dF_equiv_class}
d'_F([A]_{1:n}) = d_F(A_{1:n}).
\end{equation}
Then, $d'_F$ is an $n$-metric.
\end{theorem}

\begin{theorem}\label{th:quotient_2}
{Let $s$ be a $P$-score.
Let ${d_{\mathcal{G}_2}:\Omega^2 \mapsto \mathbb{R}}$ be the $\mathcal{G}$-align distance {function} for two graphs
induced by $s$, and  $d_\mathcal{G}:\Omega^n \mapsto \mathbb{R}$ be the $\mathcal{G}$-align distance {function} for $n$ graphs
induced by $s$.}  Let $\Omega' = \Omega \backslash \sim_{d_{\mathcal{G}_2}}$,  
and $d'_{\mathcal{G}}:\Omega'^n \mapsto \mathbb{R}$ be such that
\begin{equation}\label{eq:dG_equiv_class}
d'_{\mathcal{G}}([A]_{1:n}) = d_{\mathcal{G}}(A_{1:n}).
\end{equation}
Then, $d'_{\mathcal{G}}$ is an $n$-metric.
\end{theorem}
\vspace{-3mm}

\begin{figure*}[t!]
\centering
\includegraphics[height=3.3cm]{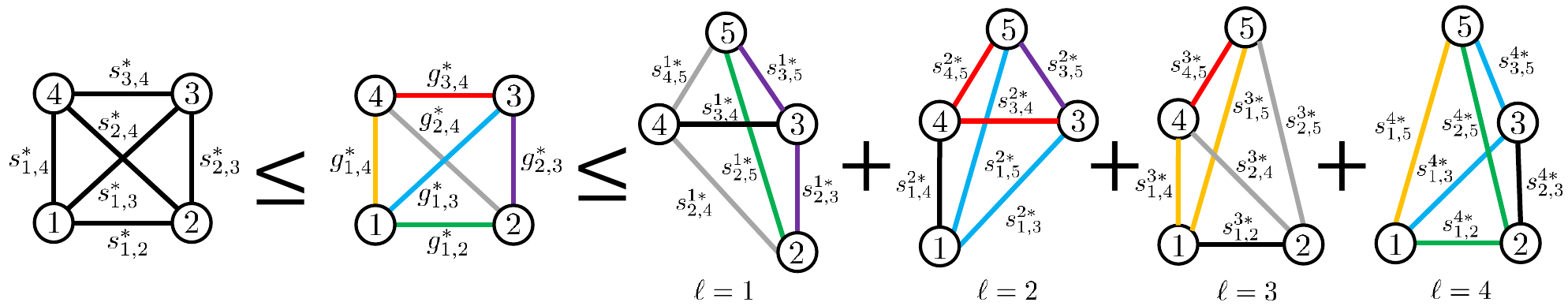}
	\caption{Generalized triangle equality of $d_\mathcal{G}$ for $n=4$ graphs.}
	\label{fig1}
\end{figure*}

\section{The generalized triangle inequality for~$d_\mathcal{G}$: an illustrative example}\label{sec:n=4}
While it is straightforward to show that $d_\mathcal{G}$ satisfies the properties of non-negativity, symmetry and {self-identity}, the proof for the generalized triangle inequality is more involved.  
To give the reader a flavor of the proof, we now prove that the $\mathcal{G}$-align function satisfies the generalized triangle inequality when $n=4$.

We consider a set of $n=4$ graphs, $\mathcal{G} = \{G_1,G_2,G_3, G_4\}$, and a reference graph $G_5$, represented by matrices, $A_1, A_2, A_3, A_4 \in \Omega$ and $A_5 \in \Omega$, respectively. 
We will show that \vspace{-2mm}
\begin{equation}\label{eq:4-1}
d_\mathcal{G}(A_{1:4}) \leq \sum_{\ell=1}^{4} d_\mathcal{G}(A^\ell_{1:4,5}).
\end{equation}
Let $P^\ast = \{P_{i,j}^{\ast}\} \in S$ be an optimal value for $P$ in the optimization problem corresponding to the left-hand-side (l.h.s) of \eqref{eq:4-1}. 
We define $s^{\ast}_{i,j} = s(A_i,A_j,P^{\ast}_{i,j})$ for all $i,j \in~[4]$. 
We also define $s^{\ell*}_{i,j} = s(A_i,A_j,P^{\ell*}_{i,j})$ for all $i,j \in [5], ~\ell\in[4]\backslash\{i,j\}$, in which $P^{\ell\ast}=\{P_{i,j}^{\ell\ast}\} \in S$ is an optimal value for 
$P$ in the optimization problem associated to $d_\mathcal{G}(A^\ell_{1:4,5})$ on the r.h.s of \eqref{eq:4-1}. {Note that, according to \eqref{s1}, and the fact that $P_{i,j}^{\ast} = {(P_{j,i}^{\ast})}^{-1}$ (since $P^{\ast} \in~S$), we have}\vspace{-1mm}
\begin{equation}
s^{\ast}_{i,j} = s^{\ast}_{j,i}, \text{ and } s^{\ell\ast}_{i,j} = s^{\ell\ast}_{j,i}. \label{eq:p2}
\end{equation}
Moreover, according to \eqref{s2}, we have
\begin{equation}
s(A_i,A_j,P_{i,k}^{\ell\ast}P_{k,j}^{\ell'\ast}) \leq s_{i,k}^{\ell\ast} +  s_{k,j}^{\ell'\ast},\label{eq:p3}
\end{equation}
and, in the particular case when $\ell = \ell'$, we have
\begin{equation}
s^{\ell\ast}_{i,j} \leq s_{i,k}^{\ell\ast} +  s_{k,j}^{\ell\ast}.\label{eq:p4}
\end{equation}
\vspace{-2mm}
From the definition of $d_\mathcal{G}$ in Lemma~\ref{th:d_G_compact_sum}, we have
\begin{equation}\label{eq:4-2}
\sum\limits_{{{i,j \in [4]},~i<j}
} s^{\ast}_{i,j} \leq
\sum\limits_{\substack{{i,j \in [4]},~i<j}
} s(A_i,A_j,\Gamma_{i,j}),
\end{equation}
where $\Gamma_{i,j} = \Gamma_i \Gamma_j^{-1}$, and $\{\Gamma_i\}$ are any set of invertible matrices in $\mathcal{P}$. Note that from Lemma~\ref{lem1},~$\{\Gamma_{i,j}\} \in S$.
Consider the following choices for $\Gamma_i$'s :
\begin{align}\label{eq:4-3}
\Gamma_1 = P_{1,5}^{4\ast};~
\Gamma_2 = P_{2,5}^{1\ast};~
\Gamma_3 = P_{3,5}^{2\ast};~
\Gamma_4 = P_{4,5}^{3\ast}.
\end{align}
We define $g_{i,j}^{\ast} = s(A_i,A_j,\Gamma_i \Gamma_j^{-1})$, in which $\Gamma_i$'s are chosen according to \eqref{eq:4-3}. We can then rewrite \eqref{eq:4-2} as
\begin{equation}\label{eq:4-3n}
\sum\limits_{{{i,j \in [4]},~i<j}
} s^{\ast}_{i,j} \leq
\sum\limits_{\substack{{i,j \in [4]},~i<j}
} g_{i,j}^{\ast}.
\end{equation}
We use Fig.~\ref{fig1} to bookkeep all the terms involved in proving \eqref{eq:4-1}. In particular, the first inequality in Fig.~\ref{fig1} provides a pictorial representation of \eqref{eq:4-3n}. In this figure, each circle represents a graph in $\mathcal{G}$, and a line between $G_i$ and $G_j$ represents the $P$-score between $A_i$ and $A_j$. In the diagram on the left, each $P$-score corresponds to the optimal pairwise matching between $G_i$ and $G_j$ associated to $d_\mathcal{G}(A_{1:4})$ in \eqref{eq:4-1}, whereas in the diagram in the middle, each $P$-score corresponds to the suboptimal matching between $G_i$ and $G_j$, where the pairwise matching matrices are chosen according to \eqref{eq:4-3}.
Using \eqref{eq:p3}, followed by \eqref{eq:p2} we get
\begin{align*}
\sum\limits_{{{i,j \in [4]},~i<j}
} \hspace{-5mm} g^{\ast}_{i,j} \leq~ 
\hspace{-2mm}&\textcolor{Green}{(s^{4\ast}_{1,5}+s^{1\ast}_{2,5})}\hspace{-1mm}+\hspace{-1mm}\textcolor{Cerulean}{(s^{4\ast}_{1,5}+s^{2\ast}_{3,5})}\hspace{-1mm}+\hspace{-1mm}\textcolor{Dandelion}{(s^{4\ast}_{1,5}+s^{3\ast}_{4,5})}+\nonumber\\[-3\jot]
\hspace{-2mm}&\textcolor{Violet}{(s^{1\ast}_{2,5}+s^{2\ast}_{3,5})}\hspace{-1mm}+\hspace{-1mm}\textcolor{Gray}{(s^{1\ast}_{2,5}+s^{3\ast}_{4,5})}\hspace{-1mm}+\hspace{-1mm}\textcolor{red}{(s^{2\ast}_{3,5}+s^{3\ast}_{4,5})}.
\end{align*}
The above inequality is also depicted in Fig.~\ref{fig1}, where each diagram on the r.h.s of the second inequality represents $d_\mathcal{G}(A^\ell_{1:5})$ in \eqref{eq:4-1} for a different $\ell \in [4]$. 
Applying \eqref{eq:p4} to the r.h.s. of the above inequality,
one can see that each one of the terms in parenthesis, distinguished with a different color, is upper bounded by the sum of the terms with the same color in 
the diagram in 
the r.h.s of 
the second inequality in Fig.~\ref{fig1}. 
This completes the proof.
\section{{Moving towards tractability}}\label{sec:tractabilty}
The following lemmas are the building blocks towards a relaxation of $d_{\mathcal{G}}$ that is also easy to compute for choices of $\mathcal{P}$ other than orthonormal matrices.  In this section, $\vvvert\cdot\vvvert_*$ denotes the nuclear norm.

\begin{lemma}\label{th:rank_relaxation}
Given $\{P_{i,j}\}_{i,j\in[n]}$ such that $P_{i,j}\in \mathbb{R}^{m\times m}$ for all $i,j\in[n]$, let ${\bf P} \in \mathbb{R}^{nm\times nm}$ have $n^2$ blocks, such that the $(i,j)$th block is $P_{i,j}$. 
Let 
\begin{align}\label{eq:low_rank_S}
&S'' = \{\{P_{i,j}\}_{i,j\in[n]}: \text{rank}({\bf P})=m, P_{i,j}\in\mathcal{P}, \forall i,j\in[{n}],\nonumber \\
&P_{i,i}=I, \forall i\in [n]\}.
\end{align}
We have that $S'' = S$, where $S$ is as defined in \eqref{eq:def_of_S}. 
\end{lemma}
\vspace{-2mm}
\begin{proof}
Let ${\bf P} \in \mathbb{R}^{nm \times nm}$, with blocks $\{P_{i,j}\}_{i,j\in[n]} \in S''$.
Since $\text{rank}({\bf P})=m$, from the singular value decomposition of ${\bf P}$,  we can write ${\bf P} = A B^\top$ where $A,B\in\mathbb{R}^{mn\times m}$. Let $A = [A_1;\dots;A_n]$, where $A_i\in \mathbb{R}^{m\times m}$ and, similarly, let $B = [B_1;\dots;B_n]$, where $B_i\in \mathbb{R}^{m\times m}$. It follows that $P_{i,j} = A_i B^\top_j$. Since $P_{i,i} = I$, we have $A_i B^\top_i = I$, which implies that 
$P_{i,j} = A_i A^{-1}_j$. By Lemma \ref{lem1}, this in turn implies that 
$\{P_{i,j}\}_{i,j\in[n]}$ satisfy the 
alignment consistency property. Therefore, $\{P_{i,j}\}_{i,j\in[n]}\in S$, and thus $S'' \subseteq S$.

Let $P = \{P_{i,j}\}_{i,j\in[n]}\in S$.
 By Lemma \ref{lem1}, $P_{i,j} = Q_i Q^{-1}_j$ for some invertible matrices $\{Q_i\}_{i\in[n]}$. Let $A,B\in\mathbb{R}^{mn\times m}$, with $A = [Q_1;\dots;Q_n]$ and $B = [(Q^{-1}_1)^\top,\dots,(Q^{-1}_n)^\top]$.
  Let ${\bf P}$ denote the $mn\times mn$ block matrix with $P_{i,j}$ as the $(i,j)$th block. We have ${\bf P} = AB^\top$. Thus $m \geq \text{rank}({\bf P}) \geq \text{rank}(A) \geq \text{rank}(Q_1)=m$,  which implies that $\{P_{i,j}\}_{i,j\in[n]}\in S''$, and therefore $S \subseteq S''$.
\end{proof}

\begin{lemma}\label{lm:pos_def_for_perms}\emph{[\hspace{-0.4mm}\cite{huang2013consistent}, Proposition 1]}
Let $\mathcal{P}$ be the set of $m\times m$ permutation matrices. Given $\{P_{i,j}\}_{i,j\in[n]}$ such that $P_{i,j}\in\mathcal{P}$ for all $i,j\in[n]$, let ${\bf P} \in \mathbb{R}^{nm\times nm}$ have $n^2$ blocks, such that the $(i,j)$th block is $P_{i,j}$. 
Let 
\begin{align}\label{eq:low_rank_S}
&S''' = \{\{P_{i,j}\}_{i,j\in[n]}: P_{i,j}\in\mathcal{P}, \forall i,j\in[{n}],{\bf P}\succeq 0, \nonumber \\
&P_{i,i}=I, \forall i\in [n]\}.
\end{align}
We have that $S''' = S$, where $S$ is as defined in \eqref{eq:def_of_S}.  
\end{lemma}
\vspace{-2mm}
\begin{lemma} \label{th:lower_bound_on_nuc_norm}
For any ${\bf P} \in \mathbb{R}^{nm\times nm}$ with ${\bf P}_{ii}=1$ for all $i\in[nm]$, we have $\vvvert{\bf P}\vvvert_* \geq nm$.
\end{lemma} 
\begin{proof}
Let ${\bf P}' = \frac{1}{2}({\bf P} + {\bf P}^\top)$. We have $nm = \text{tr}({\bf P})=\text{tr}({\bf P}') = \sum_{i \in [nm]}\lambda_i({\bf P}')\leq \sum_{i \in [nm]}|\lambda_i({\bf P}')|=\sum_{i \in [nm]}\sigma_i({\bf {\bf P}}')=\vvvert {\bf P}'\vvvert_*\leq \frac{1}{2}(\vvvert {\bf P}\vvvert_* + \vvvert {\bf P}^\top\vvvert_*) = \vvvert {\bf P}\vvvert_*$, where $\lambda_i(\cdot)$ and $\sigma_i(\cdot)$ denote the $i$th eigenvalue and the $i$th singular value of $(\cdot)$, respectively.
\end{proof}

\begin{lemma} \label{th:nuc_norm_for_orth_matrices}
Let $\mathcal{P}$ be a subset of the orthogonal matrices.
Let $\{P_{i,j}\}_{i,j\in[n]}\in S$, and ${\bf P}$ be the $mn\times nm$ block matrix with $P_{i,j}$ as the $(i,j)$th block. {We have $\vvvert {\bf P}\vvvert _* = mn$.}
\end{lemma} 
\vspace{-4mm}
\begin{proof}
Since $\{P_{i,j}\}_{i,j\in[n]}\in S$ are alignment-consistent, we can write $P_{i,j} = P_{i,n} P^{-1}_{j,n}$ for all $i,j\in[n]$. Since $P_{j,n}\in\mathcal{P}$, it must be  orthogonal. Hence, $P_{i,j} = P_{i,n} P^\top_{j,n}$,
and we can write~${\bf P} = A A^\top$, where $A = [Q_1;\dots;Q_n]\in\mathbb{R}^{nm\times m}$, and $Q_i = P_{i,n}$. Since ${\bf P}$ is positive semi-definite, its eigenvalues are equal to its singular values, which are non-negative, and thus ${\vvvert{\bf P}\vvvert_* }= \text{tr}(A A^\top)=\text{tr}(A^\top A) = \sum_{i\in[n]} \text{tr}(Q^\top_i Q_i) = \sum_{i\in[n]} \text{tr}(I)=mn$.
\end{proof}
\vspace{-2mm}
Inspired by Lemmas \ref{th:rank_relaxation}, \ref{th:lower_bound_on_nuc_norm}, and \ref{th:nuc_norm_for_orth_matrices}, to obtain a continuous relaxation of $d_{\mathcal{G}}$, we relax the rank constraint $\text{rank}({\bf P}) \leq m$ to $\vvvert {\bf P}\vvvert_*\leq mn$, use a function $s$ that is a continuous function of $P$, and use a set $\mathcal{P}$ that is compact and contains a non-empty ball around $I$. Alternatively, we can impose that $P_{j,i} = P^\top_{i,j}$, which was the case when $\mathcal{P}$ only contained orthonormal matrices, and relax the rank constraint to $\text{tr}({\bf P}) \leq mn$ and ${\bf P} \succeq 0$, i.e., ${\bf P}$ is a symmetric matrix with non-negative eigenvalues.
Note that since we want  $P_{i,i} = I$ for all $i\in [n]$, we can drop the trace constraint. The relaxation to ${\bf P} \succeq 0$ can also be justified by Lemma \ref{lm:pos_def_for_perms} and relaxing the constraint that $\mathcal{P}$ must be the set of permutations.

\begin{definition}\label{def9}
Let $\mathcal{P} \subseteq \mathbb{R}^{m \times m}$ be compact and contain a non-empty ball around $I$.
Let $P_{i,j} \in \mathcal{P}$ for all $i,j\in[n]$, and ${\bf P}$ be the $mn\times nm$ block matrix with $P_{i,j}$ as the $(i,j)$th block.
Given a map $s:\Omega^2\times \mathcal{P} \mapsto \mathbb{R}$, such that $s(\cdot,\cdot,P)$ is continuous for all $P\in \mathcal{P}$, the continuous $\mathcal{G}$-align distance function induced by $s$, is the map $d_{c\mathcal{G}}:\Omega^n \mapsto~{\mathbb{R}}$, defined by\vspace{-2mm}
\begin{equation}\label{eq:g_align_cont}
d_{c\mathcal{G}}(A_{1:n})=\hspace{-0.5cm}\min\limits_{\substack{P_{i,j} \in \mathcal{P}~\forall i,j\in[n],\\
P_{i,i}=I~\forall i\in[n],\\ \vvvert{\bf P}\vvvert_* \leq mn }}\frac{1}{2}\sum\limits_{{{i,j \in [n]}}
} s(A_i,A_j,P_{i,j}),
\end{equation}
and the symmetric continuous  $\mathcal{G}$-align distance function induced by $s$, is the map $d_{sc\mathcal{G}}:\Omega^n \mapsto{\mathbb{R}}$, defined~by
\begin{equation}\label{eq:g_align_sym_and_cont}
d_{sc\mathcal{G}}(A_{1:n})=\hspace{-0.5cm}\min\limits_{\substack{P_{i,j} \in \mathcal{P}~\forall i,j\in[n],\\
P_{i,i}=I~\forall i\in[n],\\ 
{\bf P} \succeq 0}}\frac{1}{2}\sum\limits_{{{i,j \in [n]}}
} s(A_i,A_j,P_{i,j}).
\end{equation}
\end{definition}
\vspace{-2mm}
\begin{remark}
Both optimization problems are continuous optimization problems, although they are potentially non-convex. However, for several natural choices of $s$, e.g., $s(A,B,P) = \vvvert AP-PB\vvvert$, and convex $\mathcal{P}$, both 
\eqref{eq:g_align_cont} and \eqref{eq:g_align_sym_and_cont} can be computed via convex optimization. 
\end{remark}
\vspace{-2mm}
We finish this section, by showing that the above continuous distance functions, $d_{c\mathcal{G}}$ and $d_{sc\mathcal{G}}$, are pseudo $n$-metrics. 
In what follows, we let $\|\cdot\|$ and $\vvvert \cdot \vvvert_2$ denote the Euclidean norm and matrix operator norm, respectively. We will use the following definition.
\begin{definition}\label{def:mod_s_score}
A map $s: {\Omega}^2 \times \mathcal{P} \mapsto {\mathbb{R}}$ is called a modified $P$-score, if and only if, $\mathcal{P}$ is closed under transposition and multiplication, for any $P \in \mathcal{P}$, $\vvvert P\vvvert_2 \leq 1$, and for any $P, P' \in \mathcal{P}$, and $A, B, C \in \Omega$, $s$ satisfies the properties:\vspace{-2mm}
\begin{align}
&s(A,B,P) \geq 0\label{mod:ds1},\\
&s(A,A,I) =0\label{mod:s3},\\
&{s(A,B,P) = s(B,A,P^\top)}, \label{mod:s1}\\
&s(A,B,P) + s(B,C,P') \geq s(A,C,PP')\label{mod:s2}.
\end{align}
\end{definition}
\vspace{-2mm}
\noindent For example, if $\mathcal{P}$ is the set of doubly stochastic matrices, $\Omega$ is a subset of the symmetric matrices, and $\vvvert \cdot \vvvert$ is an element-wise matrix $p$-norm, then $s(A,B,P) = \vvvert AP-BP\vvvert$ is a modified $P$-score.

We now provide the main result of this section.
\begin{theorem}\label{thm:pos_def}
If $s$ is a modified $P$-score, then the symmetric continuous $\mathcal{G}$-align distance function induced by $s$ is a pseudo $n$-metric.
\end{theorem}
\begin{remark}
A theorem with slightly different assumptions can be stated and proved about the $d_{c\mathcal{G}}$. Under appropriately defined equivalent classes, we can also obtain $n$-metrics from \eqref{eq:g_align_cont} and \eqref{eq:g_align_sym_and_cont} (cf. Section \ref{sec:quotient}).
\end{remark}

{\bf Graphs of different sizes:} We note that in this section, unlike in Sec. \ref{sec:n-metrics}, $P_{ij}$ does not need to be invertible. Therefore, it is possible to extend the (symmetric) continuous $\mathcal{G}$-align distance function to consider graphs of unequal sizes. We could, e.g., allow $P_{ij}$ to be rectangular of size $m_i$ by $m_j$ (resp. the node sizes of graph $G_i$ and $G_j$), which would still result in {$\bf P$} being square. If $P_{ij}$'s were previously doubly stochastic matrices, now, the row sums (or column sums, but not both) would be allowed to be $ \leq 1$. This would model unmatched nodes, and avoid non-trivial solutions for Eqs.~\eqref{eq:g_align_cont} and \eqref{eq:g_align_sym_and_cont}, i.e., $P_{i,j} = 0$ when $i \neq j$.

\section{Numerical experiments}\label{sec:numerics}

We do two experiments comparing our tool against two state-of-the-art non-$n$-metrics (from computer vision) and one simpler approach. Code for these comparison can be found in \url{http://github.com/bentoayr/n-metrics}. This repository includes code to compute some of our $n$-metrics, as well as code for the other methods, which is publicly available and that can be found through links in their respective papers, and which was copied into our repository for convenience.

The two competing algorithms are \emph{matchSync} \cite{pachauri2013solving}, and \emph{mOpt} \cite{yan2015consistency}. The simpler approach, \emph{Pairwise}, defines $d(G_1,...,G_{n}) = \sum_{i>j} d(G_i,G_j)$, where each $d(G_i,G_j)$, is computed using  \cite{cho2010reweighted}. All of these algorithms output a set of permutation matrices $\{P_{i,j}\}$, where $P_{i,j}$ tells how the nodes of graph $i$ and $j$ are matched. Both \emph{matchSync}, and \emph{mOpt} try to enforce the alignment consistency property on $\{P_{i,j}\}$, while \emph{Pairwise} computes each $P_{i,j}$ independently.
For our algorithm, we use \eqref{eq:g_align_cont}, with $\mathcal{P}$ being the set of doubly stochastic matrices, and $s(A,B,P) = \vvvert AP - PB \vvvert_{\text{Fro}}$. For comparison sake, after we compute $\{P_{i,j}\}$ using our algorithm, we sometimes project each $P_{i,j}$ onto the set of permutation matrices, which amounts to solving a \emph{maximum weight matching problem}.

\subsection{Multiple graph alignment experiment}\label{sec:align_exp}

We generate one Erd\"{o}s-–R\'enyi graph with edge probability $0.5$, and $7$ other graphs which are a small perturbation of the original graph (we flip edges with $0.05$ probability), such that we know the joint optimal alignment of these $n=8$ graphs, i.e. $P^*_{i,j} =I$. We then randomly permute the labels of these graphs such that the new joint optimal alignment is known but non-trivial, i.e. $P^*_{i,j}  \neq I$. We then use our $n$-metric, and the other non-$n$-metrics, to find an alignment between the graphs. Finally, we compare the alignments produced by the different methods to the optimal alignment. We repeat this $30$ times, on random instances.

For each set of permutations $\{P_{i,j}\}$ given by the different algorithms we compute the \emph{alignment quality} (AQ) and the \emph{alignment consistency} (AC).
\begin{align*}
&\text{AQ}  = 1- \frac{\sum^{n-1}_{i= 1}\sum^n_{j= i+1}\vvvert P_{i,j} -P^*_{i,j} \vvvert/2}{mn(n-1)/2},\\
& \text{AC} =1 - \frac{\sum^n_{r = 1} \sum^{n-1}_{i=1} \sum^n_{j=i+1} \vvvert P_{i,j} - P_{i,r}P_{r,j}\vvvert/2}{m n^2(n-1)/2},
\end{align*}
where $\vvvert \cdot \vvvert$ is the Frobenius norm.
We obtain the following average accuracy (over $30$ tests), and standard deviations. 
{
\begin{table}[h!]\centering
\begin{tabular}{|l|l|l|l|l|} 
\hline
 &{\bf \it Ours} &  {\bf \it mOpt} & {\bf \it matchSync} & {\bf \it Pairwise} \\ \hline
\hspace{-0.1cm}AQ\hspace{-0.1cm} & $0.94 \hspace{-0.1cm}\pm\hspace{-0.1cm} 0.01$ & $0.91\hspace{-0.1cm}\pm\hspace{-0.1cm}0.02$ & $0.90\hspace{-0.1cm}\pm\hspace{-0.1cm}0.02$ & $0.88\hspace{-0.1cm}\pm\hspace{-0.1cm}0.02\hspace{-0.1cm}$ \\ \hline
\hspace{-0.1cm}AC\hspace{-0.1cm} & $0.92 \hspace{-0.1cm}\pm\hspace{-0.1cm} 0.07$ & $1.0\hspace{-0.1cm}\pm\hspace{-0.1cm} 0.0$ & $1.0\hspace{-0.1cm}\pm\hspace{-0.1cm}0.0$ & $0.85\hspace{-0.1cm}\pm\hspace{-0.1cm}0.02\hspace{-0.1cm}$ \\ \hline
\end{tabular}
\end{table}
}
Note that,  by design, \emph{mOpt} and \emph{matchSync} have AC = 1.

In Appendix \ref{app:distribution_align_exp}, we include an histogram with the distribution of values for these two quantities.

%
%
\subsection{Graph clustering via hypergraph cut experiment}\label{sec:clust_exp}

We build two clusters of graphs, each obtained by generating (i) a Erd\"{o}s-–R\'enyi graph with edge probability $0.7$ as the cluster center, and (ii) $9$ other graphs that are a small perturbation of (i). Graphs in (ii) are generated just like in Section \ref{sec:align_exp}. We then try to recover the true clusters using different $n$-distances.

For each $n$-distance, we build a hypergraph with $20$ nodes ($1$ node per graph) and $100$ hyperedges. Each hyperedge is built by randomly connecting $3$ nodes (out of $20$), for which the distance between their graphs is below a certain threshold. This threshold is later tuned to minimize each algorithm's clustering error (define below). 
Ideally, most hyperedges should not include graphs in different clusters. We then use the algorithm of \cite{vazquez2009finding}, whose code can be found in \cite{HGC_software} and which is included in our repositories for convenience, to find a minimum cut of the hypergraph that divides it into two equal-sized parts. These hyper-subgraphs are our predicted clusters. The \emph{clustering error} is the fraction of misclassified graphs \emph{times two}, such that the worst possible algorithm, a random guess, gives an avg. error of $1$. We repeat this $50$ times. For each algorithm, we use the same threshold in all $50$ repetitions. 

This experiment does not require an alignment between graphs but only a distance $d$. For algorithms that output an alignment $\{P_{i,j}\}$, this distance is computed as $\frac{1}{2}\sum_{i,j}\vvvert A_i P_{i,j} - P_{i,j}A_j \vvvert_{\text{Fro}}$. For our algorithm, we calculate this distance by first projecting $\{P_{i,j}\}$ onto the permutation matrices, which we denote as $\emph{Ours}$, and we also calculate this distance directly as in \eqref{eq:g_align_cont}, which we denote as $\emph{Ours*}$.

We report the average error in the following table. The standard deviation of the mean are all $0.04$ except for Ours* which is $0.05$.
{
\begin{table}[h!]\centering
\begin{tabular}{|l|l|l|l|l|} 
\hline
 {\bf \it Ours*} & {\bf \it Ours} & {\bf \it mOpt} & {\bf \it matchSync} & {\bf \it Pairwise} \\ \hline
 $0.40$  &
 $0.44$  & $0.44$ & $0.49$ & $0.46$ \\\hline
\end{tabular}
\end{table}
}

In Appendix \ref{app:distribution_clust_exp} we include an histogram with the distribution of errors for the different algorithms.


\section{Future work}\label{sec:conc}

It is possible to define the notion of a (pseudo) \emph{$(C,n)$-metric}, as a map that satisfies the following more stringent generalization of the generalized triangle inequality:
%
%
$
d(A_{1:n}) \leq C\times \textstyle\sum_{i=1}^{n}  d(A^i_{1:n,n+1}).$

%
The authors in ~\cite{kiss2016generalization} prove that the $d_F$ is a (pseudo) $(C,n)$-metric with $\frac{1}{n-1} \leq C \leq \frac{1}{\lfloor{\frac{n}{2}}\rfloor}$. Any (pseudo)  $(C,n)$-metric with $C \leq 1$ is also a (pseudo)  $n$-metric. It is an open problem to determine the largest constant $C$, for which $d_\mathcal{G}$, $d_{c\mathcal{G}}$ or $d_{sc\mathcal{G}}$  are a (pseudo)  $(C,n)$-metric, and whether $C < 1$?

We also plan to test if the claim in \cite{vijayan2017pairwise}, which states that in several scenarios calculating and using pairwise alignments is better than calculating and using joint alignments, holds for the $n$-metrics we introduced.

We plan to develop fast and scalable solvers to compute our $n$-metrics. The objective function of our $n$-metrics involves a large number of sums, in turn  involving variables that are coupled by the alignment consistency constraint, or its relaxed equivalent. This makes the use of decomposition-coordination methods very attractive. In particular, we plan to test solvers based on the Alternating Direction Method of Multipliers (ADMM). Although not strictly a first -order method, it is very fast and, with proper tuning, it achieves a convergence rate that is as fast as the fastest possible first-order method \cite{francca2016explicit,nesterov2013introductory}. Furthermore, it has been used as an heuristic to solve many non-convex, even combinatorial, problems \cite{bento2013message,bento2015proximal,zoran2014shape,mathysparta}, and can be less affected by the topology of the communication network in a cluster than, e.g. Gradient Descent \cite{francca2017markov,francca2017distributed}. Finally, ADMM parallelizes well on share-memory multiprocessor systems, GPUs, and computer clusters \cite{boyd2011distributed,parikh2014block,hao2016testing}.

\newpage

\bibliography{biblio}
\bibliographystyle{plainnat}

\newpage
\clearpage

\input{appendix.tex}

\end{document}

%% file: appendix.tex
\appendix

{\bf \Large Supplementary material for ``Tractable $n$-Metrics for Multiple Graphs''}


\section{Special case of orthogonal matrices}\label{sec:ortho_matrices}
{In this section, we discuss the special case, where the pairwise matching matrices are orthogonal. This will further illustrate 
why computing $d_F$ is harder than computing $d_\mathcal{G}$.
We consider the following assumption.}
\begin{assumption} \label{ass:ortho}
$\Omega$ is the set of real symmetric matrices, namely, $\Omega=\{A\in\mathbb{R}^{m\times m}:A=A^{\top}\}$.
$\mathcal{P}$ is the set of orthogonal matrices, namely, $\mathcal{P} = \{P\in\mathbb{R}^{m\times m}:P^\top=P^{-1}\}$.
$s(A,B,P) = {\vvvert AP-PB \vvvert}~\forall A,B\in \Omega, P\in \mathcal{P}$, where $\vvvert \cdot \vvvert$ is the Frobenius norm or the operator norm, which are orthogonal invariant, and
$d(A,B) = \min_{P\in \mathcal{P}} s(A,B,P)$.
\end{assumption}

{We now provide the main results of this section in the following theorems, and provide the detailed proofs in Appendix \ref{App_thm5}-\ref{App_thm7}.}
\begin{theorem}\label{thm:5}
Under Assumption \ref{ass:ortho}, $d_F$ induced by $d$, and $d_\mathcal{G}$ induced by $s$, are pseudo $n$-metrics.
\end{theorem}
\begin{theorem}\label{th:orthogonal_form_for_df}
Let $\Lambda_{A_i}\in\mathbb{R}^m$ be the vector of eigenvalues of $A_i$, ordered from largest to smallest. Then, under Assumption~\ref{ass:ortho}, 
\begin{equation}
d_F(A_{1:n}) = \min_{\Lambda_C\in\mathbb{R}^m} \sum^n_{i=1} \|\Lambda_{A_i}-\Lambda_C\|.\label{eq:df_for_orth}
\end{equation}
%
\end{theorem}
\begin{theorem}\label{thm:7}
Let $\Lambda_{A_i}\in\mathbb{R}^m$ be the vector of eigenvalues of~$A_i$, ordered from largest to smallest. Then, under Assumption~\ref{ass:ortho}, 
\begin{equation}
d_\mathcal{G}(A_{1:n}) =  \frac{1}{2}\sum_{i,j\in[n]} \|\Lambda_{A_i}-\Lambda_{A_j}\|.\label{eq:dG_for_orth}
\end{equation}
%
\end{theorem}
\vspace{-4mm}
{Note that $d_F = d_\mathcal{G} = 0$ if and only if $A_{1:n}$ {share the same spectrum}. 
}

\vspace{-2mm}
{The function $d_F$ is related to the geometric median of the spectra of $A_{1:n}$. In order to write \eqref{eq:dG_for_orth} as an optimization problem similar to $d_F$ in \eqref{eq:df_for_orth}}, it is tempting to define $d_\mathcal{G}$ using $s^2$ instead of $s$, and take a square root.
Let us call the resulting function $\bar{d}_\mathcal{G}$.
A straightforward calculation allows us to
write \vspace{-2mm}
\begin{align*}
&(\bar{d}_\mathcal{G}(A_{1:n}))^2=\frac{1}{2}\sum_{i,j\in[n]} \|\Lambda_{A_i}-\Lambda_{A_j}\|^2\\
&= n^2 \left(\frac{1}{n}\sum_{i\in[n]}\Big\|\Lambda_{A_i} -  \frac{1}{n}\sum_{j\in[n] }\Lambda_{A_j}\Big\|^2\right)\equiv n^2\text{Var}(\Lambda_{A_{1:n}})\nonumber\\
&=n \min_{\Lambda_C\in\mathbb{R}^m} \frac{1}{2}\sum_{i\in[n]} \|\Lambda_{A_i}-\Lambda_C\|^2,
\end{align*}
where we use $\text{Var}(\Lambda_{A_{1:n}})$  to denote the geometric sample variance of the vectors $\{\Lambda_{A_i}\}$. This leads to a definition very close to \eqref{eq:df_for_orth}, and a connection between $\bar{d}_\mathcal{G}$ and the geometric sample variance.
%

At this point it is important to note that sample variances can be computed exactly in $\mathcal{O}(n)$ steps involving only sums and products of numbers. Contrastingly, although there are fast approximation algorithms for the geometric median \cite{cohen2016geometric}, there are no procedures to compute it exactly in a finite number of simple algebraic operations \cite{bajaj1986proving,cockayne1969euclidean}.
%


%
%

\section{Proof of Theorem \ref{thm1}}\label{sec:proof_thm1}
In the following lemmas, we show that the Fermat distance function satisfies properties~\eqref{h1},~\eqref{h3},~\eqref{h4}, and {\eqref{eq:self_identity}}, and hence is a pseudo $n$-metric. 
\begin{lemma}
$d_F$ is non-negative.
\end{lemma}
\begin{proof}
If $d$ is a {pseudo} metric, it is non-negative. Thus,~\eqref{eq:fermat} is the sum of non-negative functions, and hence also non-negative.
\end{proof}
\begin{lemma}
	{$d_F$ satisfies the self-identity property.}
\end{lemma}
\begin{proof}
{If $A_1=A_2=\ldots=A_n$, then $d_F(A_{1:n}) = \min\limits_{B} n \times d(A_1,B)$, which is zero if we choose $B=A_1\in\Omega$, and \eqref{eq:self_identity} follows}.
\end{proof}
\begin{lemma}
	$d_F$ is symmetric.
\end{lemma}
\begin{proof}
	Property \eqref{h3} simply follows from the commutative property of summation.
\end{proof}
\begin{lemma}\label{lem6}
	$d_F$ satisfies the generalized triangle inequality.
\end{lemma}
\begin{proof}
Note that the following proof is a direct adaptation of the one in~\cite{kiss2016generalization}, and is included for the sake of completeness. 
We show that the Fermat distance satisfies \eqref{h4}, i.e.,
\begin{equation}
d_F(A_{1:n}) \leq \sum_{i=1}^{n} d_F(A^i_{1:n,n+1}).
\end{equation}
Consider $B_{1:n} \in \Omega$ such that,
\begin{equation}\label{eq:fermat_tri_2}
d_F(A^i_{1:n,n+1}) =  d(A_{n+1}, B_i) + \sum_{j \in [n]\setminus i} d(A_j, B_i).
\end{equation}
Equation~\eqref{eq:fermat_tri_2} implies that 
\begin{eqnarray}\label{eq:fermat_triangle}
\sum_{i=1}^{n} d_F(A^i_{1:n,n+1}) &\geq& \sum_{i=1}^{n} \sum_{\substack{j \in [n]\backslash i}} d(A_j, B_i) \geq d(A_1,B_n)+\nonumber\\
d(A_2,B_n)&+&\sum_{i=2}^{n-1}(d(A_1,B_i)+d(A_{i+1},B_i)).
\end{eqnarray}
Using triangle inequality, we have $d(A_1,B_n)+
d(A_2,B_n) \geq d(A_1,A_2)$, and, 
$d(A_1,B_i)+d(A_{i+1},B_i) \geq d(A_1,A_{i+1})$. Thus, from \eqref{eq:fermat_triangle},
\begin{eqnarray*}
\sum_{i=1}^{n} d_F(A^i_{1:n,n+1}) \geq \hspace{-1.5mm} \sum_{i=2}^{n}d(A_1,A_i)=\hspace{-1.5mm} \sum_{i=1}^{n}d(A_1,A_i)
\geq& \hspace{-3mm} d_F(A_{1:n}),
\end{eqnarray*}
where we used $d(A_1,A_1)=0$ in the equality. The last inequality follows from Definition~\ref{def6}, and completes the proof.
\end{proof}
\section{Proof of Theorem \ref{thm2}}\label{sec:proof_thm2}
In the following lemmas, we show that the $\mathcal{G}$-align distance function satisfies properties~\eqref{h1},~\eqref{h3},~\eqref{h4}, and {\eqref{eq:self_identity}}, and hence is a pseudo $n$-metric. 

\begin{lemma}
	$d_\mathcal{G}$ is non-negative.
\end{lemma}
\begin{proof}
Since $s$ is a $P$-score, it satisfies {\eqref{ds1}, i.e., $s \geq 0$}, which implies $d_\mathcal{G} \geq 0$, since it is a sum of $P$-scores.
\end{proof}
\begin{lemma}
	{$d_\mathcal{G}$ satisfies the self-identity property}.
\end{lemma}
\begin{proof}
If $A_1 = A_2  =\ldots = A_n$, then, if we choose $P\in S$ such that $P_{i,j} = I$ for all $i,j \in [n]$, we have  $s(A_i,A_j,P_{i,j})=0$ by \eqref{s3}, for all $i,j \in [n]$. Therefore, 
\begin{equation*}
0 \leq d_\mathcal{G}(A_{1:n}) \leq \frac{1}{2}\sum_{i,j\in[n]}s(A_i,A_j,P_{i,j}) =0.\vspace{-7mm}
\end{equation*}
\end{proof}
\begin{lemma}
	$d_\mathcal{G}$ is symmetric.
\end{lemma}
\begin{proof}
	The definition, \eqref{eq:galign},  involves summing $s(A_i,A_j, P_{i,j})$ over all pairs $i, j \in [n]$, which clearly makes $d_\mathcal{G}$ invariant to permuting $\{A_i\}$.
\end{proof}

\begin{lemma}
	$d_\mathcal{G}$ satisfies the generalized triangle inequality.
\end{lemma}
\begin{proof}

We now show that $d_\mathcal{G}$ satisfies \eqref{h4}, i.e.,
\begin{equation}\label{eq:G_tri_1}
d_\mathcal{G}(A_{1:n}) \leq \sum_{\ell=1}^{n} d_\mathcal{G}(A^i_{1:n,n+1}).
\end{equation}

Let $P^\ast = \{P_{i,j}^{\ast}\} \in S$ be an optimal value for $P$ in the optimization problem corresponding to the l.h.s of~\eqref{eq:G_tri_1}. 
Henceforth, just like Section \ref{sec:n=4}, we {use} $s^{\ast}_{i,j} = s(A_i,A_j,P^{\ast}_{i,j})$ for all $i,j \in~[n]$.
Note that according to \eqref{s3} and \eqref{s1}, we have $s^{\ast}_{i,i} = 0$, and $s^{\ast}_{i,j} = s^{\ast}_{j,i}$, respectively.
From \eqref{eq:th:d_G_compact_sum}, we have,
\begin{equation}\label{eq:G_tri_2}
d_\mathcal{G}(A_{1:n}) = \sum\limits_{{i,j \in [n],~i<j}
} s(A_i,A_j,P_{i,j}^{\ast}) = \sum\limits_{{i,j \in [n],~i<j}
} s^{\ast}_{i,j}.
\end{equation}

Let $P^{k\ast} =\{ P_{i,j}^{k\ast}\}\in S$ be an optimal value for 
$P$ in the optimization problem associated to $d_\mathcal{G}(A^i_{1:n,n+1})$ on the r.h.s of~\eqref{eq:G_tri_1}.
{Henceforth, just like Section \ref{sec:n=4}, we use} $s^{\ell*}_{i,j} = s(A_i,A_j,P^{\ell*}_{i,j})$ for all $i,j \in [n+1], \ell\in[n]\backslash\{i,j\}$.
Note that $s^{\ell\ast}_{i,i}= 0$, and $s^{\ell\ast}_{i,j} = s^{\ell\ast}_{j,i}$.
From \eqref{eq:th:d_G_compact_sum}, we can write,
\begin{equation}\label{eq:G_tri_3}
\sum_{\ell=1}^{n} d_\mathcal{G}(A^i_{1:n,n+1}) = \sum\limits_{\ell=1}^{n} \sum\limits_{\substack{{i,j \in [n+1],~{i<j}}\\\ell \notin \{i,j\}}}
s^{\ell\ast}_{i,j}.
\end{equation}
We will show that,
\begin{equation}\label{eq:G_tri_4}
\sum\limits_{\substack{{i,j \in [n],~i<j}}
} s^{\ast}_{i,j} \leq
\sum\limits_{\ell=1}^{n} \sum\limits_{\substack{{i,j \in [n+1],~{i<j}}\\\ell \notin \{i,j\}}} s^{\ell*}_{i,j}.
\end{equation}
From the definition of $d_\mathcal{G}$ in Lemma~\ref{th:d_G_compact_sum},
\begin{equation}\label{eq:G_tri_5}
\sum\limits_{{{i,j \in [n]},~i<j}
} s^{\ast}_{i,j} \leq
\sum\limits_{\substack{{i,j \in [n]},~i<j}
} s(A_i,A_j,\Gamma_{i,j}),
\end{equation}
for any matrices $\{\Gamma_{i,j}\}_{i,j \in [n]}$ in $S$, where $S$ satisfies Definition \ref{def7}.
Hence, from Lemma~\ref{lem1}, we also know that
\begin{equation}\label{eq:G_tri_6}
\sum\limits_{\substack{{i,j \in [n]},~i<j}} s^{\ast}_{i,j} \leq
\sum\limits_{\substack{{i,j \in [n]},~i<j}} s(A_i,A_j,\Gamma_{i} \Gamma_{j}^{-1}),
\end{equation}
for any invertible matrices $\{\Gamma_{i}\}_{i \in [n]}$ in $\mathcal{P}$.

Consider the following choice for $\Gamma_i$ :
\begin{align}
&\Gamma_i = P_{i,n+1}^{i-1\ast},\qquad 2\leq i \leq n, \label{eq:G_tri_7}\\
&\Gamma_1 = 
P_{1,n+1}^{n\ast}.\label{eq:i1}
\end{align}

\begin{remark}\label{rmk:i_modulus_n}
To simplify notation, we will just use $\Gamma_i = P_{i,n+1}^{i-1\ast}$ for all $i\in[n]$. It is assumed that when we writing
$P_{i,j}^{\ell\ast}$, the index in superscript satisfies $\ell=0 \Leftrightarrow \ell=n$.
\end{remark}

Note that since {$P^{i-1*} \in S$}, then $\Gamma_i=P_{i,n+1}^{i-1\ast}$ is invertible and belongs to $\mathcal{P}$.
Using \eqref{eq:G_tri_7} to replace $\Gamma_{i}$ and $\Gamma_j$ in \eqref{eq:G_tri_6}, and the fact that $(P_{j,n+1}^{j-1\ast})^{-1} = {P_{n+1,j}^{j-1\ast}}$, along with property \eqref{s2} of the $P$-score $s$, we have
\begin{eqnarray*}\label{eq:G_tri_8}
\sum\limits_{\substack{{i,j \in [n]}\\i<j}} s(A_i,A_j,\Gamma_{i} \Gamma_{j}^{-1})&=&\sum\limits_{\substack{{i,j \in [n]}\\i<j}} s(A_i,A_j,P_{i,n+1}^{i-1\ast} {P_{n+1,j}^{j-1\ast}})\\
&\leq&
\sum\limits_{\substack{{i,j \in [n]}\\i<j}} 
s^{i-1\ast}_{i,n+1} + s^{j-1\ast}_{n+1,j}.
\end{eqnarray*}
We now show that
\begin{align}\label{eq:G_tri_9}
&\sum\limits_{\substack{{i,j \in [n]}\\i<j}} s^{i-1\ast}_{i,n+1} + s^{j-1\ast}_{n+1,j} \leq \sum\limits_{\ell=1}^{n} \sum\limits_{\substack{i,j \in [n+1],~i<j\\{\ell \notin \{i,j\}}}} s^{\ell\ast}_{i,j},
\end{align}
which will prove \eqref{eq:G_tri_4} and complete the proof of the generalized triangle inequality for $d_\mathcal{G}$.

To this end, let $I_1 = \{(i,j)\in[n]^2: i<j,~j-1 = i\}$, 
$I_2 = \{(i,j)\in[n]^2:i=1,j=n\}$, 
$I_3 = \{(i,j)\in[n]^2:i<j,~j-1  \neq i \text{ and } (i,j) \neq (1,n)\}$.
We will make use of the following three inequalities, which follow directly from  property \eqref{s2} of the $P$-score $s$.
\begin{align}
&\sum\limits_{\substack{{(i,j)\in I_1 }}} s^{i-1\ast}_{i,n+1} \leq \sum\limits_{\substack{{(i,j)\in I_1 }}} s^{i-1\ast}_{i,j} + s^{i-1\ast}_{j,n+1}.\label{eq:G_tri_10}\\
&\sum\limits_{\substack{{(i,j)\in I_2 }}} s^{j-1\ast}_{n+1,j} \leq \sum\limits_{\substack{{(i,j)\in I_2 }}} s^{j-1\ast}_{n+1,i} + s^{j-1\ast}_{i,j}.\label{eq:G_tri_11}\\
&\sum\limits_{\substack{{(i,j)\in I_3}}} 
s^{i-1\ast}_{i,n+1} + s^{j-1\ast}_{n+1,j} \leq
\sum\limits_{\substack{{(i,j)\in I_3 }}} \Big(s^{i-1\ast}_{i,j} + s^{i-1\ast}_{j,n+1} +\nonumber\\
&~~s^{j-1\ast}_{n+1,i} + s^{j-1\ast}_{i,j}\Big).\label{eq:G_tri_12}
\end{align}
Since $I_1$, $I_2$ and  $I_3$ are pairwise disjoint, we have
\begin{align}\label{eq:G_tri_13}
&~~\sum\limits_{{{i,j \in [n]}}} (\cdot) =
\sum\limits_{{{(i,j)\in I_1}}} (\cdot)+\sum\limits_{{{(i,j)\in I_2}}} (\cdot)+\sum\limits_{{{(i,j)\in I_3}}}(\cdot).
\end{align}
Using  \eqref{eq:G_tri_10}-\eqref{eq:G_tri_12}, and \eqref{eq:G_tri_13} we have
\begin{align}\label{eq:G_tri_14}
&\hspace{-0.5mm} \sum\limits_{\substack{{i,j \in [n]},~i<j}} \hspace{-4mm} s^{i-1\ast}_{i,n+1} + s^{j-1\ast}_{n+1,j}\leq \hspace{-2mm} \sum\limits_{\substack{{(i,j)\in I_1 }}} \hspace{-2mm} s^{i-1\ast}_{i,j} + s^{i-1\ast}_{j,n+1} + s^{j-1\ast}_{n+1,j}+\nonumber\\
&~~\sum\limits_{\substack{{(i,j)\in I_2 }}} s^{i-1\ast}_{i,n+1} + s^{j-1\ast}_{n+1,i} + s^{j-1\ast}_{i,j}+\nonumber\\
&~~\sum\limits_{\substack{{(i,j)\in I_3 }}} s^{i-1\ast}_{i,j} + s^{i-1\ast}_{j,n+1} + s^{j-1\ast}_{n+1,i} + s^{j-1\ast}_{i,j}.
\end{align}
%
%
%
To complete the proof, we show that the r.h.s of \eqref{eq:G_tri_14} is less than, or equal to\vspace{-4mm}
\begin{equation}\label{eq:G_tri_15}
\sum\limits_{\ell=1}^{n} \sum\limits_{\substack{{i,j \in [n+1],~i<j}\\{\ell \notin \{i,j\}}}} s_{i,j}^{\ell\ast}.
\end{equation}
To establish this, we show that each term on the r.h.s of \eqref{eq:G_tri_14} is: (i) not repeated; and (ii) is included in \eqref{eq:G_tri_15}.

\begin{definition}
We call two $P$-scores, $s_{a_1,b_1}^{c_1\ast}$ and $ s_{a_2,b_2}^{c_2\ast}$, \emph{coincident}, and denote it by 
$s_{a_1,b_1}^{c_1\ast} \sim s_{a_2,b_2}^{c_2\ast}$,  if and only if $c_1 = c_2$, \emph{and} $\{a_1,b_1\}=\{a_2,b_2\}$.
\end{definition}
\vspace{-1mm}
Checking (i) amounts to verifying that there are no  coincident terms on the r.h.s. of \eqref{eq:G_tri_14}. Checking (ii) amounts to verifying that for each $P$-score $s_{a_1,b_1}^{c_1\ast}$ on the r.h.s. of \eqref{eq:G_tri_14}, there exists a $P$-score 
$s_{a_2,b_2}^{c_2\ast}$ in \eqref{eq:G_tri_15} such that ~$s_{a_1,b_1}^{c_1\ast} \sim s_{a_2,b_2}^{c_2\ast}$.

Note that the r.h.s of \eqref{eq:G_tri_14} consists of three summations. 
To verify (i), we first compare the terms within each summation, and then compare the terms among different summations.
Consider the first summation on the r.h.s of \eqref{eq:G_tri_14}. We have $s^{i-1\ast}_{i,j} \not \sim s^{i-1\ast}_{j,n+1}$ because $i \in [n]$ and therefore $i \neq n+1$. We have 
$s^{i-1\ast}_{i,j} \not \sim s^{j-1\ast}_{n+1,j}$ because $i-1 \neq j-1$ in this case, since $i<j$. We can similarly infer that $s^{i-1\ast}_{j,n+1} \not \sim s^{j-1\ast}_{n+1,j}$. 

Now consider the second summation on the r.h.s of \eqref{eq:G_tri_14}. 
Taking the definition of $I_2$ and \eqref{eq:i1} into account, we can rewrite this summation as,
\begin{align}\label{eq:G_tri_16}
s^{n\ast}_{1,n+1} + s^{n-1\ast}_{n+1,1} + s^{n-1\ast}_{1,n}.
\end{align}
Since $n \neq n-1$, we have  $s^{n\ast}_{1,n+1} \not \sim s^{n-1\ast}_{n+1,1}$, and $s^{n\ast}_{1,n+1} \not \sim s^{n-1\ast}_{1,n}$. Also, since $n \neq n+1$ we have $s^{n-1\ast}_{n+1,1} \not \sim s^{n-1\ast}_{1,n}$.

Finally, consider the third summation on the r.h.s of \eqref{eq:G_tri_14}. Since $i<j$, by comparing the superscripts we immediately see that the first and second terms in the summation cannot be equal to either the third or the forth term. On the other hand, since $n+1\neq i\in[n]$ and $ n+1 \neq j \in[n]$, we have
$s^{i-1\ast}_{i,j} \not \sim s^{i-1\ast}_{j,n+1}$, and
$s^{j-1\ast}_{n+1,i} \not \sim s^{j-1\ast}_{i,j}$, respectively. 

We proceed by showing that  the summands are not coincident among three summations. We first make the following observations:

{\bf{Observation $\bf1$}}: since in all summations $i,j \in [n]$, we have $i \neq n+1,~j \neq n+1$, and therefore 
each term with $n+1$ in the subscript is not coincident with any  term with $\{i,j\}$ in the subscript, e.g., on the r.h.s of \eqref{eq:G_tri_14}, the first terms in the first and second summations cannot be coincident.

{\bf{Observation $\bf2$}}: since $I_1$, $I_2$ and $I_3$ are pairwise disjoint, any two terms from different summations with the same indices cannot be coincident, e.g.,  on the
r.h.s  of \eqref{eq:G_tri_14}, the third term in the second summation cannot be coincident with the third term in third summation.

Considering the above observations, the number of pairs we need to compare reduces from $3\times 7+3\times 4=33$ (in \eqref{eq:G_tri_14}) pairs to only $13$ pairs, whose distinction may not seem trivial. To be specific, Obs. 1, excludes $16$ comparisons and Obs. 2 excludes $4$ comparisons. We now rewrite the r.h.s of 
\eqref{eq:G_tri_14} as
\begin{align}\label{eq:G_tri_17}
&~\sum\limits_{\substack{{(i,j)\in I_1 }}}~s^{i-1\ast}_{i,j} + s^{i-1\ast}_{j,n+1} + s^{j-1\ast}_{n+1,j}+\nonumber\\[-3\jot]
&~~~~~~~~~~~s^{n\ast}_{1,n+1} + s^{n-1\ast}_{n+1,1} + s^{n-1\ast}_{1,n}+\nonumber\\
&~\hspace{-0.5mm}\sum\limits_{\substack{{(i',j') \in I_3 }}} s^{i'-1\ast}_{i',j'} + s^{i'-1\ast}_{j',n+1} + s^{j'-1\ast}_{n+1,i'} + s^{j'-1\ast}_{i',j'}.
\end{align}
\vspace{-1mm}
In what follows, we discuss the non-trivial comparisons, and refer to the first, second and third summations in \eqref{eq:G_tri_17} as $\Sigma_1$, $\Sigma_2$, and $\Sigma_3$, respectively.
\begin{enumerate}
\item $s^{i-1\ast}_{i,j}$ in $\Sigma_1$ vs. $s^{n-1\ast}_{1,n}$ in $\Sigma_2$: for these two terms to be coincident we need $i=n$. We also need $\{n,j\}=\{1,n\}$, i.e., $j=1$, which cannot be true, since in $S_1$ we have $i=j-1$ according to~$I_1$.
\item {$s^{i-1\ast}_{i,j}$ in $\Sigma_1$ vs. $s^{j'-1\ast}_{i',j'}$ in $\Sigma_3$:} since $(i,j) \in I_1 = \{(i,j)\in[n]^2: i<j,~j-1 = i\}$, we have $j=i+1$. Thus, we can write the first term as $s^{i-1\ast}_{i,i+1}$. For the two terms to be coincident, their superscripts must be the same so 
 $i=j'$. On the other hand, for their subscripts to match, we need $j =i+1= i'$. The last two equalities imply that $i' = j'+1$, which contradicts $(i',j')\in I_3$.
\item $s^{i-1\ast}_{j,n+1}$ in $\Sigma_1$ vs. $s^{n\ast}_{1,n+1}$ in $\Sigma_2$: for the superscripts to match, we need $i=1$. We also need $j=1$ for the equality of subscripts, which cannot be true since $i < j$.
\item $s^{i-1\ast}_{j,n+1}$ in $\Sigma_1$ vs.
$s^{n-1\ast}_{n+1,1}$ in $\Sigma_2$: we need $i=n$ for the equality of superscripts, and $j=1$ for the equality of subscripts, which cannot be true since $(i,j) \in I_1$, and therefore $i=j-1$.
\item $s^{i-1\ast}_{j,n+1}$ in $S_1$ vs. $s^{j'-1\ast}_{n+1,i'}$ in $S_3$:  we can write the first term as $s_{i+1,n+1}^{i-1}$. 
The equality of superscripts requires $i=j'$. The equality of subscripts requires $i'=i+1$. Therefore, $i'=j'+1$, which contradicts $(i',j')\in I_3$.
\item $s^{j-1\ast}_{n+1,j}$ in $\Sigma_1$ vs.
$s^{n\ast}_{1,n+1}$ in $\Sigma_2$: the equality of superscripts requires $j=1$, which is impossible since $j > i \in [n]$.
\item $s^{j-1\ast}_{n+1,j}$ in $\Sigma_1$ vs.
$s^{n-1\ast}_{n+1,1}$ in $\Sigma_2$: for the equality of superscripts, we need $j=n$, in which case the subscripts will not match, since $\{n+1,n\}\neq\{n+1,1\}$.
\item $s^{j-1\ast}_{n+1,j}$ in $\Sigma_1$ vs.
$s^{i'-1\ast}_{j',n+1}$ in $\Sigma_3$: the equality of superscripts requires $i'=j$. The equality of the subscripts requires $j'=j$. The two equalities imply that $i' = j'$, which contradicts $i'<j'$.
\item $s^{j-1\ast}_{n+1,j}$ in $\Sigma_1$ vs.
$s^{j'-1\ast}_{n+1,i'}$ in $\Sigma_3$: the equality of superscripts requires $j'=j$. The equality of the subscripts requires  $i'=j$. The two equalities imply that  $i' = j'$, which contradicts $i'<j'$.
\item $s^{n\ast}_{1,n+1}$ in $\Sigma_2$ vs. $s^{i'-1\ast}_{j',n+1}$ in $\Sigma_3$: for the equality of superscripts, we need $i'=1$, and for the equality of subscripts, we need $j'=1$. This contradicts $i' < j'$.
\item $s^{n\ast}_{1,n+1}$ in $\Sigma_2$ vs.
$s^{j'-1\ast}_{n+1,i'}$ in the $\Sigma_3$: for equality of superscripts, we need $j'=1$. For the equality of subscripts, we need $i'=1$, which contradicts $i'\neq j'$.
\item $s^{n-1\ast}_{n+1,1}$ in $\Sigma_2$ vs. 
$s^{i'-1\ast}_{j',n+1}$ in $\Sigma_3$: for equality of superscripts, we need $i'=n$. For the equality of subscripts, we need $j'=1$, which contradicts $i' < j'$.
\item $s^{n-1\ast}_{1,n}$ in $\Sigma_2$ vs.
$s^{i'-1\ast}_{i',j'}$ in $\Sigma_3$: for the equality of superscripts, we need $i'=n$. This in turn requires $j'=1$ for the equality of subscripts, which contradicts $i'<j'$.
\end{enumerate}
%
%
What is left to show is (ii), i.e., that all terms in \eqref{eq:G_tri_17} are included in the summation in \eqref{eq:G_tri_15}. To this aim, we will show that for each $s_{a,b}^{c}$ in \eqref{eq:G_tri_17}, the indices $\{a,b,c\}$ satisfy 
\vspace{-2mm}
\begin{equation}\label{eq:proof_prop_ii}
c \in [n], a,b \in [n+1]\backslash\{c\} \text{ and } a \neq b,
\end{equation}
which is enough to prove that either $s_{a,b}^{c}$ or $s_{b,a}^{c}$ exist in~\eqref{eq:G_tri_15}.

We first note that the superscripts in \eqref{eq:G_tri_17} are in $[n]$, see Remark \ref{rmk:i_modulus_n}.
Moreover, all the subscripts in \eqref{eq:G_tri_17} are either $1$, $n+1$, or $i,j,i',j' \in [n]$.
Thus, for any $s_{a,b}^{c}$ in \eqref{eq:G_tri_17}, we have $a,b \in [n+1]$.
Also note that, for any $s_{a,b}^{c}$ in \eqref{eq:G_tri_17}, we have $a \neq b$, since the definition of $I_1$, $I_2$ and $I_3$ implies that $i<j$, $i'<j'$ and $i,j,i',j'< n+1$. Therefore, all we need to verify is that for any $s_{a,b}^{c}$ in \eqref{eq:G_tri_17}, $a \neq c$ and $b \neq c$.

We start with the first summation, where the first term is $s_{i,j}^{i-1\ast}$. Clearly $i \neq i-1$ and $j \neq i-1$, from the definition of~$I_1$. In the second term, $s_{j,n+1}^{i-1\ast}$, $j \neq i-1$, from the definition of $I_1$, and $i-1 \neq n+1$, because otherwise 
$i=n+2\notin [n]$. In the third term, $s_{n+1,j}^{j-1\ast}$, we have $n+1 \neq j \in [n]$. Moreover, clearly $j \neq j-1$.

For any term $s_{a,b}^{c}$, in the second summation, we clearly see in \eqref{eq:G_tri_17} that $a \neq c$ and $b \neq c$. 

We now consider the last summation in \eqref{eq:G_tri_17}. In the first term, $s^{i'-1\ast}_{i',j'}$, clearly $i' \neq i'-1$. Moreover, $i'-1<i'<j'$, since $(i',j')\in I_3$. 
In the second term, $s^{i'-1\ast}_{j',n+1}$, $j' \neq i'-1$,
because since $i'-1<i'<j'$. Moreover, $n+1 \neq i'-1$ because otherwise $i'=n+2 \not\in [n]$.
In the third term, $s^{j'-1\ast}_{n+1,i'}$, we have $n+1 \neq j'-1$ because otherwise $ j' = n+2 \not \in [n]$. On the other hand, $i' \neq j'-1$ since $(i',j') \in I_3$.
In the fourth term, $s^{j'-1\ast}_{i',j'}$, we have $i' \neq j'-1$ since $(i',j') \in I_3$. Also, clearly $j' \neq j'-1$. 
\end{proof}
%
%
%
%
\section{Proof of Theorem~\ref{thm:5}}\label{App_thm5}
To show that $d_F$ is a pseudo $n$-metric, it suffices to show that $d$ is a pseudometric, and evoke Theorem \ref{thm1}. To show that $d$ is a pseudometric, we can evoke {Theorem $3$}
in~\cite{bento2018family}.

To show that $d_\mathcal{G}$ is a pseudo $n$-metric, it suffices to show that $s$ is a $P$-score, and evoke Theorem~\ref{thm2}. Clearly, $s$ is non-negative, and also $s(A,A,I)=0$. 
Recall that, if $P$ is orthogonal
then, for any matrix $M$, we have {$\vvvert PM \vvvert=\vvvert MP \vvvert=\vvvert M \vvvert$}. Thus,
\begin{align*}
s(A,B,P) &= \vvvert AP-PB \vvvert=\vvvert P^{-1}(AP-PB)P^{-1} \vvvert\\
&=\vvvert P^{-1}A-BP^{-1} \vvvert=s(B,A,P^{-1}).
\end{align*}
Finally, for any $P,P'\in\mathcal{P}$,
\begin{align*}
&s(A,B,PP')=\vvvert APP'-PP'B \vvvert = \\
&\vvvert APP'-PCP'+PCP'-PP'B \vvvert \leq \\ &\vvvert APP'-PCP' \vvvert+\vvvert PCP'-PP'B \vvvert =\\
&\vvvert AP-PC \vvvert+\vvvert CP'-P'B \vvvert=\\
&s(A,C,P)+s(C,B,P'). \qedhere
\end{align*}
\vspace{-4mm}
%
%
%
\vspace{-0.6cm}
\section{Proof of Theorem~\ref{th:orthogonal_form_for_df}}\label{App_thm6}
The proof uses the following lemmas by \cite{hoffman1953variation} and \cite{bento2018family}.
\begin{lemma}\label{th:ortho_matrix_norm_mult} For any matrix $M\in\mathbb{R}^{m\times m}$, 
and any orthogonal matrix $P\in\mathbb{R}^{m \times m}$,
{we have that $\vvvert PM \vvvert = \vvvert MP \vvvert = \vvvert M \vvvert$.}
\end{lemma}
\begin{lemma}\label{th:aux_lemma_eigen_values}
Let $\vvvert \cdot \vvvert$ be the Frobenius norm.
If $A$ and $B$ are Hermitian matrices
with eigenvalues $a_1 \leq a_2 \leq ...\leq a_m$ and $b_1 \leq b_2 \leq ...\leq b_m$ then 
\begin{equation}
{\vvvert A-B \vvvert} \geq \sqrt{{\sum_{i \in [m]} (a_i - b_i)^2}}.
\end{equation}
\end{lemma}
\begin{lemma}\label{th:aux_lemma_eigen_values_op_norm}
Let $\vvvert \cdot \vvvert$ be the operator $2$-norm.
If $A$ and $B$ are Hermitian matrices
with eigenvalues $a_1 \leq a_2 \leq ...\leq a_m$ and $b_1 \leq b_2 \leq ...\leq b_m$ then 
\begin{equation}
{\vvvert A-B \vvvert} \geq \max_{i\in[m]} |a_i - b_i|.
\end{equation}
\end{lemma}
\vspace{-2mm}
We also need the following result.\vspace{-2mm}
\begin{corollary}\label{th:better_to_sort_first}
If $a\in\mathbb{R}^m$, with $a_1\leq a_2 \leq \dots \leq a_m$, $b\in\mathbb{R}^m$, with $b_1\leq b_2 \leq \dots \leq b_m$, and $P\in \mathbb{R}^{m\times m}$ is a permutation matrix, then
\begin{equation}
\|a - b\| \leq \|a - Pb\|.
\end{equation}
%
\end{corollary}
\begin{proof}
This follows directly from Lemma \ref{th:aux_lemma_eigen_values} and Lemma \ref{th:aux_lemma_eigen_values_op_norm} by letting $A$ and $B$ be diagonal matrices with  $a$ and $Pb$ in the diagonal, respectively.\vspace{-1mm}
\end{proof}

We now proceed with the proof of Theorem \ref{th:orthogonal_form_for_df}.
Let $A_i = U_i \text{diag}(\Lambda_{A_i}) U^{-1}_i$ and $C = V \text{diag}(\Lambda_{C}) V^{-1}$ be the eigendecomposition of the real and symmetric matrices $A_i$ and $C$, respectively. 
The eigenvalues in the  vectors $\Lambda_{A_i} $ and $\Lambda_{C} $ are ordered in increasing order, and $U_i$ and $V$ are orthonormal matrices.
Using Lemma \ref{th:ortho_matrix_norm_mult}, we have that
\begin{align}\label{eq:proof_for_ortho_for_df}
&{\vvvert A_i P_{i} - P_{i} C \vvvert}=\vvvert (A_i  - P_{i} C (P_{i})^{-1})P_{i} \vvvert\\
&=\vvvert A_i  - P_{i} C (P_{i})^{-1} \vvvert \nonumber\\
&=\vvvert U_i(\text{diag}(\Lambda_{A_i})  - U^{-1}_i P_{i} C (P_{i})^{-1}U_i)U^{-1}_i \vvvert \nonumber\\
&=\vvvert \text{diag}(\Lambda_{A_i})  - U^{-1}_i P_{i} C (P_{i})^{-1}U_i \vvvert \geq {\| \Lambda_{A_i} - \Lambda_C \|} \nonumber,
\end{align}
where the last inequality follows from Lemma \ref{th:aux_lemma_eigen_values} or Lemma \ref{th:aux_lemma_eigen_values_op_norm} (depending on the norm). 

It follows from \eqref{eq:proof_for_ortho_for_df} that 
$
d_F(A_{1:n}) \geq \min_{\Lambda_C\in\mathbb{R}^m:(\Lambda_C)_i \leq (\Lambda_C)_{i+1}} \sum^n_{i=1} \|\Lambda_{A_i}-\Lambda_C\|=\min_{\Lambda_C\in\mathbb{R}^m} \sum^n_{i=1} \|\Lambda_{A_i}-\Lambda_C\|$,
where the last equality follows from Corollary \ref{th:better_to_sort_first}.

Finally, notice that, by the equalities in \eqref{eq:proof_for_ortho_for_df}, we have
\begin{align}
&d_F(A_{1:n}) = \min_{P\in\mathcal{P}^n,C\in\Omega} \sum^n_{i=1} {\vvvert \text{diag}(\Lambda_{A_i})  - U^{-1}_i P_{i} C (P_{i})^{-1}U_i\vvvert}\nonumber\\
&\leq {\vvvert\text{diag}(\Lambda_{A_i}) - \text{diag}(\Lambda_{C})\vvvert},
\end{align}
where the inequality follows from upper bounding $\min_{C\in\Omega} (\cdot)$ with the particular choice of 
$C=P_{i}^{\top}U_i\text{diag}(\Lambda_{C})U^{\top}_i P_{i} \in\Omega$.

{Since $\vvvert\text{diag}(\Lambda_{A_i}) - \text{diag}(\Lambda_{C})\vvvert_{\text{Frobenius}} = \|\Lambda_{A_i} - \Lambda_{C}\|_{\text{Eucledian}}$ and $\vvvert \text{diag}(\Lambda_{A_i}) - \text{diag}(\Lambda_{C})\vvvert_{\text{operator}} = \|\Lambda_{A_i} - \Lambda_{C}\|_{\text{$\infty$-norm}}$, the proof follows.}
%
%
%
%
%
\section{Proof of Theorem~\ref{thm:7}}\label{App_thm7}
Let $A_i = U_i \text{diag}(\Lambda_{A_i}) U^{-1}_i$ be the eigendecomposition of the real and symmetric matrix $A_i$. 
The eigenvalues in the  vector $\Lambda_{A_i} $ are ordered in increasing order, and $U_i$ is an orthonormal matrix.
Using Lemma \ref{th:ortho_matrix_norm_mult}, we get
\begin{align}\label{eq:proof_for_ortho_for_d_G}
&{\vvvert A_i P_{i,j} - P_{i,j} A_j \vvvert=\vvvert (A_i  - P_{i,j} A_j (P_{i,j})^{-1})P_{i,j} \vvvert}\\
&=\vvvert A_i  - P_{i,j} A_j (P_{i,j})^{-1} \vvvert \nonumber\\
&=\vvvert U_i(\text{diag}(\Lambda_{A_i})  - U^{-1}_i P_{i,j} A_j (P_{i,j})^{-1}U_i)U^{-1}_i \vvvert\nonumber\\
&=\vvvert\text{diag}(\Lambda_{A_i})  - U^{-1}_i P_{i,j} A_j(P_{i,j})^{-1}U_i\vvvert \hspace{-1mm}\geq \hspace{-1mm} \|\Lambda_{A_i} - \Lambda_{A_j}\|\nonumber,
\end{align}
where the last inequality follows from Lemma \ref{th:aux_lemma_eigen_values} or Lemma \ref{th:aux_lemma_eigen_values_op_norm} (depending on the norm). 

From \eqref{eq:proof_for_ortho_for_d_G} we have $d_\mathcal{G}(A_{1:n}) \geq  \frac{1}{2}\sum_{i,j\in[n]} \|\Lambda_{A_i}-\Lambda_{A_j}\|$.

At the same time, \textcolor{black}{$d_{\mathcal{G}}(A_{1:n}) = $}
\begin{align}
&\min_{P\in S} \frac{1}{2}\hspace{-1mm}\sum_{i,j\in[n]} \hspace{-1mm} \vvvert \text{diag}(\Lambda_{A_i})  - U^{-1}_i P_{i,j} A_j (P_{i,j})^{-1}U_i \vvvert\nonumber\\
&\leq {\vvvert\text{diag}(\Lambda_{A_i}) - \text{diag}(\Lambda_{A_j})\vvvert},
\end{align}
where the inequality follows from upper bounding $\min_{P\in S} (\cdot)$ by choosing
$P=\{P_{i,j}\}_{i,j\in[n]}$ such that $P_{i,j} = U_i U^{-1}_j$, which by Lemma \ref{lem1} implies that $P\in S$.

{Since $\vvvert\text{diag}(\Lambda_{A_i}) - \text{diag}(\Lambda_{A_j})\vvvert_{\text{Frobenius}} = \|\Lambda_{A_i} - \Lambda_{A_j}\|_{\text{Eucledian}}$ and $\vvvert \text{diag}(\Lambda_{A_i}) - \text{diag}(\Lambda_{A_j})\vvvert_{\text{operator}} = \|\Lambda_{A_i} - \Lambda_{A_j}\|_{\text{$\infty$-norm}}$, the proof follows.}

\section{Proof of Theorem \ref{th:quotient_1}}
\label{app:proofs_for_quotient_spaces_1}

{We first show that \eqref{eq:dF_equiv_class} is well defined.}
Let~$A'_i \in [A_i]$. Since  $d$ satisfies the triangle inequality
we have
\begin{align*}
&d'_F([A']_{1:n})=d_F(A'_{1:n}) = \min_{B\in {\Omega}} \sum_{i\in[n]}d(A'_i,B) \\
&\leq \min_{B\in \Omega} \sum_{i\in[n]}d(A'_i,A_i) + d(A_i,B) \nonumber = \min_{B\in \Omega}  \sum_{i\in[n]} d(A_i,B) \\
&{=d_F(A_{1:n})}=d'_F([A]_{1:n}),
\end{align*}
where in the last equality we used $d(A'_i,A_i)=0$, since $A'_i \in [A_i]$. Similarly, we can show that $d'_F([A]_{1:n}) \leq d'_F([A']_{1:n})$. It follows that $d'_F([A]_{1:n}) = d'_F([A']_{1:n})$, and hence \eqref{eq:dF_equiv_class} is well defined.

We now prove that $d'_F$ satisfies \eqref{h2-n}. Recall that, by Theorem \ref{thm1}, $d_F$ is a pseudo $n$-metric.
If $[A_1]=\dots=[A_n]$, then 
\begin{equation*}
d'_F([A]_{1:n})=d'_F([A_1],\dots,[A_1]) = d_F(A_1,\dots,A_1)=0,
\end{equation*}
%
%
since, $d_F$ is a pseudometric, and hence satisfies the property of self-identity \eqref{eq:self_identity}.

On the other hand, if $d'_F([A]_{1:n})=d_F(A_{1:n})=0$, then there exists $B\in\Omega$, such that $d(A_i,B)=0$ for all $i\in[n]$. Since $d$ is non-negative and symmetric, and also satisfies the triangle inequality, it follows that
\begin{align*}
0 \leq d(A_i,A_j) &\leq d(A_i,B) + d(B,A_j)\\
&= d(A_i,B) + d(A_j,B) = 0. 
\end{align*}
%
Hence, $[A_i] = [A_j]$ for all $i,j\in[n]$.
%


\section{Proof of Theorem \ref{th:quotient_2}}
\label{app:proofs_for_quotient_spaces_2}

In the proof,  {we let $S_2$ denote the set $S$ in definition \eqref{eq:def_of_S} for the distance $d$ on two graphs and we let $S_n$ denote the set $S$ in definition \eqref{eq:def_of_S} for the distance $d_\mathcal{G}$ on $n$ graphs.}

We first verify that \eqref{eq:dG_equiv_class} is well defined. Let $A'_i\in[A_i]$. Let $\{I,P^*_i,(P^*_i)^{-1}\}\in S_2$ be such that \vspace{-2mm}
\begin{align*}
{d_{\mathcal{G}_2}(A_i,A'_i)}\equiv \frac{1}{2}(&s(A_i,A_i,I)\hspace{-1mm}+\hspace{-1mm}s(A'_i,A'_i,I)+\nonumber\\
&s(A'_i,A_i,P^*_i)\hspace{-1mm}+\hspace{-1mm}s(A_i,A'_i,(P^*_i)^{-1}))=0.
\end{align*}

\vspace{-3mm}
Since $s$ is a $P$-score, $s(A'_i,A_i,P^*_i)=0$. 
For any $\tilde{P}=\{\tilde{P}_{i,j}\}_{i,j\in[n]}\in S$ we have $\{ P^*_i \tilde{P}_{i,j}(P^*_j)^{-1}\}_{i,j\in[n]} \in S$. Thus,
\vspace{-2mm}
\begin{align*}
d'_\mathcal{G}([A']_{1:n})&=d_\mathcal{G}(A'_{1:n})=\min_{P \in S}\frac{1}{2} \sum_{i,j\in[n]} s(A'_i,A'_j,P_{i,j}) \nonumber\\
&\leq \frac{1}{2} \sum_{i,j\in[n]} s(A'_i,A'_j, P^*_i \tilde{P}_{i,j}(P^*_j)^{-1}).
\end{align*}
By property \eqref{s2} and the fact that $s(A'_i,A_i, P^*_i )=s(A_i,A'_i, (P^*_i)^{-1} )=0$ for all $i\in[n]$, we can write
\begin{align*}
&\frac{1}{2} \hspace{-1mm} \sum_{i,j\in[n]} \hspace{-1mm} s(A'_i,A'_j, P^*_i \tilde{P}_{i,j}(P^*_j)^{-1}) \leq \frac{1}{2} \hspace{-1mm} \sum_{i,j\in[n]} \hspace{-1mm} \Big(s(A'_i,A_i, P^*_i )\nonumber\\
&+s(A_i,A_j, \tilde{P}_{i,j})+s(A_j,A'_j, (P^*_j)^{-1} \Big)=s(A_i,A_j, \tilde{P}_{i,j})\nonumber. 
\end{align*}
%

Taking the minimum of the r.h.s. of the above expression over $\tilde{P}$ we get $d'_\mathcal{G}([A']_{1:n}) \leq d_\mathcal{G}(A_{1:n}) =d'_\mathcal{G}([A]_{1:n}).$
%
Similarly, we can prove $d'_\mathcal{G}([A]_{1:n}) \leq d'_\mathcal{G}([A']_{1:n})$. It follows that $d'_\mathcal{G}([A]_{1:n}) = d'_\mathcal{G}([A']_{1:n})$, and hence \eqref{eq:dG_equiv_class} is well defined.

Now we show that $d'_{\mathcal{G}}$ satisfies \eqref{h2-n}.
Recall that, by Thm.~\ref{thm2}, $d_\mathcal{G}$ is a pseudo $n$-metric.
If $[A_1]=\dots=[A_n]$, then 
\begin{align*}
d'_G([A]_{1:n})=d'_\mathcal{G}([A_1],\dots,[A_1]) = d_\mathcal{G}(A_1,\dots,A_1)=0,
\end{align*}
since, $d_\mathcal{G}$ is a pseudometric, and hence satisfies the property of self-identity \eqref{eq:self_identity}.

On the other hand, if $d'_G([A]_{1:n})=d_\mathcal{G}(A_{1:n})=0$, then, for any $i,j\in[n]$, we have that $s(A_i,A_j,P_{i,j})=0$ for some $P_{i,j}$, and hence $d(A_i,A_j)=0$. This implies that $[A_i]=[A_j]$ for all $i,j\in~[n]$.

\section{Proof of Theorem \ref{thm:pos_def}}

The following lemma will be used later.
\begin{lemma}\label{lm:towards_tract}
Let $\Gamma_i \in \mathbb{R}^{m  \times m}$, $\vvvert \Gamma_i \vvvert_2 \leq 1$ for all $i \in [n]$. Let ${\bf P} \in \mathbb{R}^{nm\times nm}$ have $n^2$ blocks such that the $(i,j)$th block is $\Gamma_i\Gamma^\top_j$ if $i \neq j$, and $I$ otherwise. We have that ${\bf P} \succeq 0$, and that $\vvvert {\bf P} \vvvert_* \leq mn$.
\end{lemma}
\vspace{-2mm}
\begin{proof}
Let us first prove that ${\bf P} \succeq 0$.
Let ${\bf v} \in \mathbb{R}^{nm}$ have $n$ blocks, the $i$th block being $v_i \in \mathbb{R}^m$. Since $\vvvert \Gamma_i \Gamma^\top_i  \vvvert_2 \leq \vvvert \Gamma_i\vvvert_2 \vvvert \Gamma^\top_i  \vvvert_2 \leq 1$, we have that $\|\Gamma_i^\top v_i\|^2_2 = \| v^\top _i\Gamma_i \Gamma^\top_i  v_i\|_2 \leq \|v_i\|^2_2$ for all $i \in[n]$. Therefore, we
have ${\bf v}^\top {\bf P} {\bf v} = \|\sum_{i \in [n]} \Gamma^\top_i v_i\|^2_2 + \sum_{i \in [n]} \|v_i\|^2_2 - \sum_{i \in [n]} \|\Gamma_i^\top v_i\|^2_2 \geq 0$, for any ${\bf v}$, which implies that ${\bf P } \succeq 0$.
We now prove that $\vvvert {\bf P} \vvvert_* \leq mn$. Let $\sigma_r$ and $\lambda_r$ be the $r$th singular value and $r$th eigenvalue of ${\bf P}$ respectively. Since $ {\bf P}$ is real-symmetric and positive semi-definite, we have that $\vvvert {\bf P} \vvvert_* = \sum_r \sigma_r = \sum_r |\lambda_r| = \sum_r \lambda_r = \text{tr}({\bf P}) = mn$.
\end{proof}
\vspace{-2mm}

{\it Proof of Theorem \ref{thm:pos_def}.}

(Non-negativity): Since $s$ is a modified $P$-score, it satisfies {\eqref{mod:ds1}, i.e., $s \geq 0$}, which implies $d_{sc\mathcal{G}} \geq 0$, since the objective function on the r.h.s of \eqref{eq:g_align_sym_and_cont} is a sum of modified $P$-scores.

(Self-identity): If $A_1 = A_2  =\ldots = A_n$, then, if we choose $P_{i,j} = I$ for all $i,j \in [n]$, we have  $s(A_i,A_j,P_{i,j})=0$ by \eqref{mod:s3}, for all $i,j \in [n]$. Note that from the definition of $d_{sc\mathcal{G}}$, we are assuming that $I \in \mathcal{P}$. Furthermore, ${\bf P}$ defined using these $P_{i,j}$'s satisfies $\vvvert {\bf P }\vvvert_* \leq mn$. Therefore, this choice of $P_{i,j}$'s satisfies the constraints in the minimization problem in the  definition of $d_{sc\mathcal{G}}(A_{1:n})$. Therefore, $d_{sc\mathcal{G}}(A_{1:n})$ is upper-bounded by $0$, which along with its non-negativity leads to $d_{sc\mathcal{G}}(A_{1:n})=0$.

(Symmetry): The optimization problem in \eqref{eq:g_align_sym_and_cont}, involves summing $s(A_i,A_j, P_{i,j})$ over all pairs $i, j \in [n]$. Thus, permuting the matrices $\{A_i\}$ is the same as solving \eqref{eq:g_align_sym_and_cont} with $P_{i,j}$ replaced by $P_{\sigma(i),\sigma(j)}$ for some permutation $\sigma$. Thus, all that we need to show is that ${\bf P} \succeq 0$ if and only if ${\bf P}'  \succeq 0$, where ${\bf P}'$ is just like ${\bf P}$ but with its blocks' indexes permuted. To see this, note that the eigenvalues of a matrix $M$ do not change if $M$ is then permuted under some permutation matrix~$T$. 

(Generalized triangle inequality): We will follow exactly the same argument as in the proof of the generalized triangle inequality for Theorem~\ref{thm2}, which is provided in Appendix~\ref{sec:proof_thm2}. The only modification is in equation \eqref{eq:G_tri_6}, and in a couple of steps afterwards. 

Equation \eqref{eq:G_tri_6} should be replaced with 
\vspace{-2mm}
\begin{equation} \label{eq:neq_G_tri_6}
\sum\limits_{i \neq j} s(A_i,A_j,P^{\ast}_{i,j}) \leq \sum\limits_{i \neq j} s(A_i,A_j,\Gamma_{i} \Gamma_{j}^{\top}),
\end{equation}
where $\{\Gamma_{i}\}_{i \in [n]}$ are matrices in $\mathcal{P}$. This inequality holds because $P_{i,j}$ defined by $P_{i,j} = \Gamma_i \Gamma_j^\top$ $\forall i\neq j$, and  $P_{i,i} = I \forall i$, satisfies the constraints in \eqref{eq:g_align_cont}, and hence the r.h.s. of \eqref{eq:neq_G_tri_6} upper bounds the optimal objective value for \eqref{eq:g_align_cont}. Indeed, since $\Gamma_i \in \mathcal{P}$, and since, by assumption, $\mathcal{P}$ is closed under multiplication and transposition, it follows that $\Gamma_i \Gamma_j^\top\in \mathcal{P}$. Furthermore, if we define ${\bf P}$ to have as the $(i,j)$th block, $i\neq j$, $\Gamma_i \Gamma_j^\top$, and have as the $(i,i)$th block the identity $I$, then, by Lemma \ref{lm:towards_tract}, we know that ${\bf P}\succeq 0$.

Starting from \eqref{eq:neq_G_tri_6}, we use \eqref{mod:s2} and \eqref{mod:s1} from the modified $P$-score properties and obtain
\begin{align}
&\sum\limits_{i \neq j} s(A_i,A_j,\Gamma_{i} \Gamma_{j}^{\top}) \leq
\sum\limits_{i \neq j} s(A_i,A_{n+1},\Gamma_{i}) +\nonumber\\
&\sum\limits_{i \neq j} s(A_{n+1},A_j,\Gamma_{j}^{\top}) = s(A_i,A_{n+1},\Gamma_{i})
\\
&+\sum\limits_{i \neq j} s(A_j,A_{n+1},\Gamma_{j}).
\end{align}
The rest of the proof follows by choosing $\Gamma_i$ has in \eqref{eq:G_tri_7} and \eqref{eq:i1}, and noting that the new definition of $s^*_{i,j}$ and $s^{\ell*}_{i,j}$ satisfies the same properties as in the proof of Theorem \ref{thm2}. In particular, we have that $s^*_{i,j} =s^*_{j,i}$ and $s^{\ell*}_{i,j} = s^{\ell*}_{j,i}$, because ${\bf P}$ in \eqref{eq:g_align_sym_and_cont} is symmetric, and because we are assuming that \eqref{mod:s1} holds.
%

%
%
\section{Distribution of AQ and AC for the alignment experiment} \label{app:distribution_align_exp}

\begin{figure}[h!]
\includegraphics[trim=2.0cm 6.0cm 1.5cm 5.5cm, clip=true,width=\linewidth]{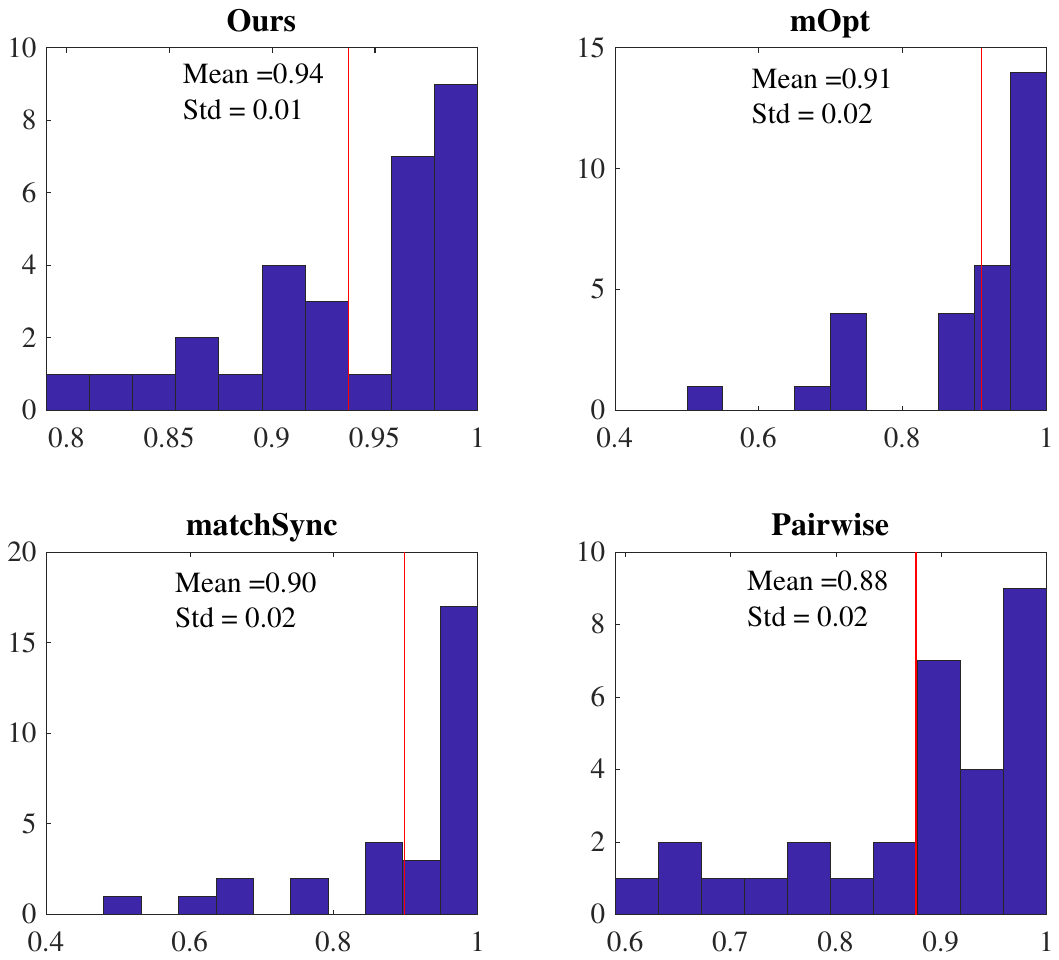}
\caption{Distribution of alignment quality (AQ) for the $30$ tests in Section \ref{sec:align_exp}.}
\vspace{-0.7cm}
\end{figure}
%
%
\begin{figure}[h!]
\includegraphics[trim=2.0cm 10.cm 1.5cm 9.cm, clip=true,width=\linewidth]{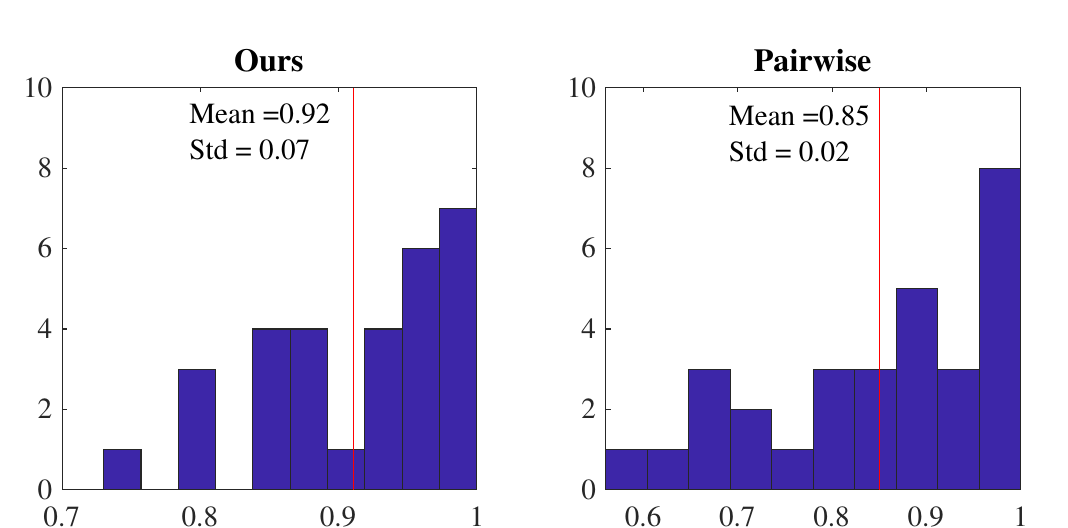}
\caption{Distribution of alignment consistency (AC) for the $30$ tests in Section \ref{sec:align_exp}. Note that, by construction, \emph{mOpt} and \emph{matchSync} always have AC = 1.}
\end{figure}

%
%
\section{Distribution of clustering errors for the clustering experiment} \label{app:distribution_clust_exp}
\begin{figure}[h!]
\includegraphics[trim=2.0cm 6.cm 1.5cm 5.cm, clip=true,width=\linewidth]{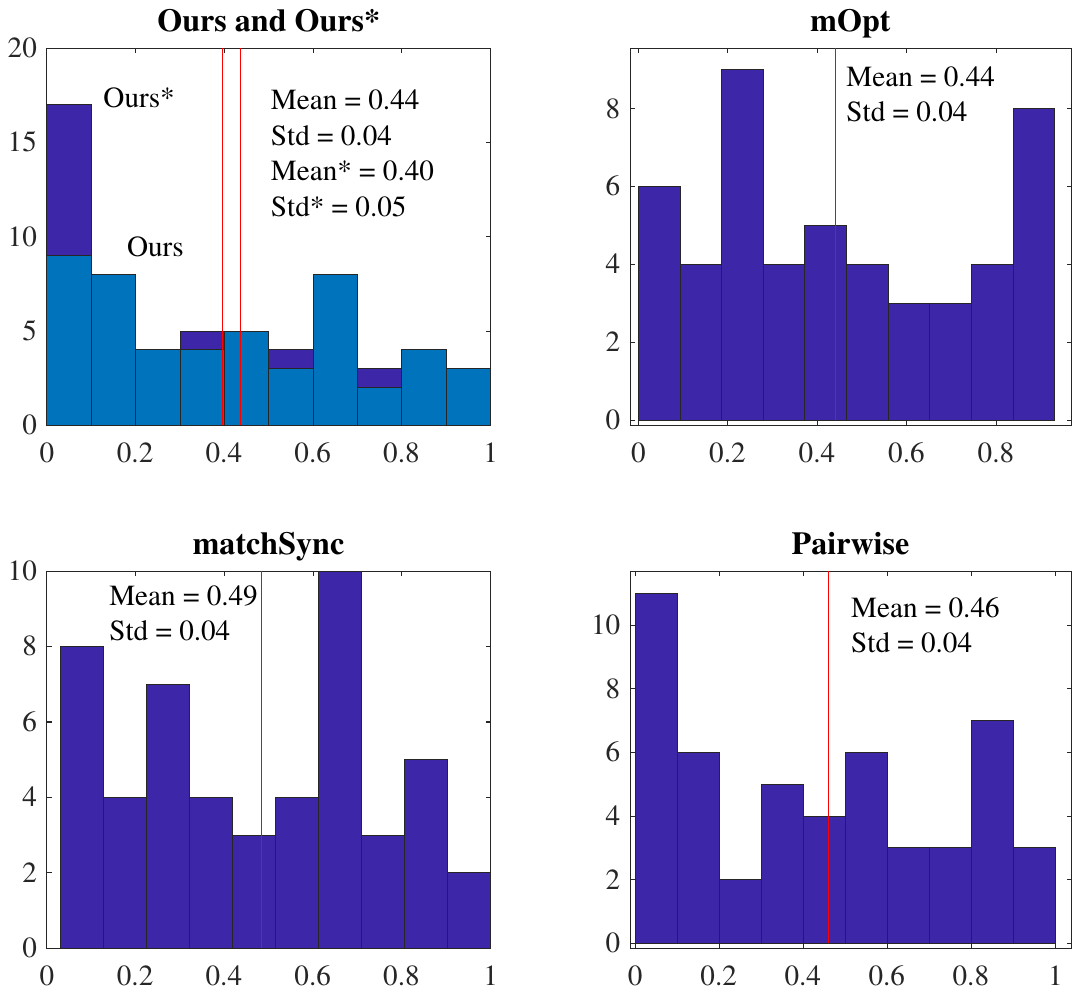}
\caption{Distribution of errors for clustering for the $50$ tests in Section \ref{sec:clust_exp}. Recall that the error is the fraction of misclassified graphs times the number of clusters, which is $2$ in our case.
A random guess gives an average clustering error of $1$.}
\end{figure}